\documentclass[12pt, draftclsnofoot, onecolumn]{IEEEtran}

\usepackage{amsfonts}
\usepackage{amssymb}
\usepackage{stfloats}
\usepackage{graphicx}
\usepackage{subfigure}
\usepackage{amsmath}
\usepackage{array}
\usepackage{epstopdf}
\usepackage{graphicx}
\usepackage{amsthm}
\usepackage{algorithm}
\usepackage{algpseudocode}
\usepackage{amsmath}
\usepackage{graphics}
\usepackage{epsfig}
\usepackage{cite}

\usepackage[usenames]{color}

\newtheorem{Definition}{Definition}
\newtheorem{Problem}{Problem}
\newtheorem{Lemma}{Lemma}

\newtheorem{Algorithm}{Algorithm}
\newtheorem{Policy}{Policy}

\newtheorem{Remark}{Remark}

\newtheorem{Example}{Example}

\begin{document}

\title{Scheduling for Mobile Edge Computing with Random User Arrivals-- An Approximate MDP and Reinforcement Learning Approach}

\author{
	\IEEEauthorblockN {Shanfeng Huang\IEEEauthorrefmark{1}\IEEEauthorrefmark{2}, Bojie Lv\IEEEauthorrefmark{1}\IEEEauthorrefmark{3}, Rui Wang\IEEEauthorrefmark{1}}\IEEEauthorrefmark{3}, Kaibin Huang\IEEEauthorrefmark{2}

    \IEEEauthorblockA{
	\IEEEauthorrefmark{1}Southern University of Science and Technology, Shenzhen, China\\
	\IEEEauthorrefmark{2}The University of Hong Kong, Hong Kong, China\\
	\IEEEauthorrefmark{3}Peng Cheng Laboratory, Shenzhen, China}
	\thanks{Part of this work has been accepted by the IEEE Global Communications Conference 2019 \cite{huang2019mecranduser}. We have extended the conference version substantially by improving the baseline policy to achieve a better performance in Section IV-A,
			proposing a novel reinforcement learning algorithm to evaluate the value function of the baseline policy without system statistics in Section V-A, devising a stochastic gradient descent algorithm to optimize the baseline policy in Section V-B,  and generating more illustrative simulation results. 
		
This work has been submitted to the IEEE journal for possible publication.
Copyright may be transferred without notice, after which this version may
no longer be accessible.	
}
}

\maketitle

\begin{abstract}
	In this paper, we investigate the scheduling design of a mobile edge computing (MEC) system, where active mobile devices with computation tasks randomly appear in a cell. Every task can be computed at either the mobile device or the MEC server. We jointly optimize the task offloading decision, uplink transmission device selection and power allocation by formulating the problem as an infinite-horizon Markov decision process (MDP). Compared with most of the existing literature, this is the first attempt to address the transmission and computation optimization with the random device arrivals in an infinite time horizon to our best knowledge. Due to the uncertainty in the device number and location, the conventional approximate MDP approaches addressing the curse of dimensionality cannot be applied. An alternative and suitable low-complexity solution framework is proposed in this work. We first introduce a baseline scheduling policy, whose value function can be derived analytically  with the statistics of random mobile device arrivals. Then, one-step policy iteration is adopted to obtain a sub-optimal scheduling policy whose performance can be bounded analytically. The complexity of deriving the sub-optimal policy is reduced dramatically compared with conventional solutions of MDP by eliminating the complicated value iteration. To address a more general scenario where the statistics of random mobile device arrivals are unknown, a novel and efficient algorithm integrating reinforcement learning and stochastic gradient descent (SGD) is proposed to improve the system performance in an online manner. Simulation results show that the gain of the sub-optimal policy over various benchmarks is significant.
\end{abstract}

\IEEEpeerreviewmaketitle

\section{Introduction}
The last decade has witnessed an unprecedented increase in mobile data traffic. In the meanwhile,  new mobile applications with intensive computation tasks and stringent latency requirements, such as face recognition, online gaming and mobile augmented reality are gaining popularity. Due to the limited battery lives and computing capabilities of mobile devices, some computation-intensive tasks need to be offloaded to more powerful edge servers, which necessitates new network architecture design. Mobile edge computing (MEC) is an emerging architecture where cloud computing capabilities are extended to the edge of the cellular networks, in close proximity to mobile users\cite{Abbas2018MECSurvey}. MEC is envisioned as a promising solution to easing the conflict between resource-hungry applications and resource-limited mobile devices \cite{beck2016whitepaper}. In this paper, we shall investigate the joint transmission and computation scheduling in an MEC system with random user arrivals via novel approaches of approximate MDP and reinforcement learning. 

\subsection{Related Works}
Resource management of MEC systems has been intensively investigated in recent years. In \cite{You2015SingleUserWPT}, the authors considered a single-user MEC system powered by wireless energy transfer. The closed-form expression of offloading decision, local CPU frequency and time division between wireless energy transfer and offloading were derived via convex optimization theory. The authors in \cite{you2016multiuser} extended the work to a multi-user scenario and formulated the multi-user resource allocation problem as a convex optimization problem where an insightful threshold-based optimal offloading strategy was derived. Moreover, game-theory-based algorithms were designed to resolve the contention in multi-user MEC offloading decision problems in \cite{chen2016gametheorymec,chen2015decentalizedgame}.

In the above works, the dynamics (arrival or departure) of mobile devices are ignored. Moreover, they assume the transmission and computation of a task can be finished within one physical-layer frame, which may not be the case in many applications. Considering the randomness of channel fading and task arrivals, the scheduling in MEC systems becomes a stochastic optimization problem. A number of research attempts have been devoted to such scheduling problems in MEC systems. In \cite{huang2012dynamic}, the authors considered a single-user MEC system and proposed a Lyapunov optimization algorithm to minimize the long-term average energy consumption. The authors in \cite{mao2016power-delay} investigated the power-delay tradeoff of a multi-user MEC system via Lyapunov optimization. Also, the authors in \cite{liu2016delayopt} solved the power-constrained delay-optimal task scheduling problem for an MEC system via MDP. Moreover, the authors in \cite{Ko2018HetnetMEC} proposed a spatial and temporal computation offloading decision algorithm in edge cloud-enabled heterogeneous networks via MDP, where multiple users and multiple computation nodes were considered. Additionally, with the popularity of artificial intelligence, a bunch of recent works on the scheduling of MEC systems have come forth leveraging the tool of deep reinforcement learning \cite{XQiu2019DRL-MEC,LTan2018DRL-Cache-MEC,JWang2019ResorceAlloc-MEC,YLiu2019DRL-VehicleEdge,YHe2018cache-comp-DRL}. Nevertheless, all these works consider the resource management with either a single mobile device or a number of fixed mobile devices. The scheduling design with random arrivals of mobile devices remains open.

In addition to MEC scheduling, MDP has been widely used in various resource allocation problems of wireless communication systems. For example, the delay-aware radio resource management for uplink, downlink and cooperative systems has been investigated in \cite{YCui2010AMDP-OFDM-uplink, YCui2010AMDP-OFDM-downlink,RWang2011DistTowHopMIMO,RWang2013RelayApproxMDP,Han2016,Han2018,Han2019}, where approximate MDP is usually adopted to address the curse of dimensionality. However, all these works considered the wireless communication scenarios with fixed transmitters and receivers. The approximate MDP approaches developed in these works as well as the deep reinforcement learning methods used in \cite{XQiu2019DRL-MEC,LTan2018DRL-Cache-MEC,JWang2019ResorceAlloc-MEC,YLiu2019DRL-VehicleEdge,YHe2018cache-comp-DRL} can not be directly applied to solve the problems considering the randomness of mobile devices in both number and locations. Moreover, these methods lack of sufficient design insights and there are no analytical performance bounds on the proposed algorithms.

\subsection{Motivations and Contributions}
 As mentioned above, existing works mainly consider the scenarios where either single or multiple mobile devices at fixed locations offload computation tasks via uplink. The arrival of new offloading mobile devices or the departure of existing ones is excluded in these scenarios. In practice, when the computation task of a mobile device is finished, the mobile device may become inactive and new devices with computation tasks may join the system in a stochastic manner. To the best of our knowledge, the resource optimization in MEC systems with random user arrivals remains largely untapped.

In this paper, we would like to shed some lights on the above issue by optimizing the task offloading in a cell with random mobile device (task) arrivals in both temporal and spatial domains. Specifically, active mobile devices with one computation task arrives randomly, and their locations follow certain spatial distribution. The tasks can be computed either locally or remotely at the edge server (via uplink). The joint optimization of task offloading decision, uplink device selection and power allocation in all the frames is formulated as an infinite-horizon MDP with discounted cost. Our main contributions on this new scheduling problem are summarized below.

\begin{itemize}
	\item \textbf{A novel low-complexity approximate MDP framework:} 
	Due to the dynamics of user arrival and departure, the number of mobile devices in the MEC system is variable. The system state space should enumerate all the possible numbers of mobile devices. The conventional approximate MDP approaches in \cite{YCui2010AMDP-OFDM-downlink,YCui2010AMDP-OFDM-uplink,RWang2011DistTowHopMIMO, RWang2013RelayApproxMDP,Han2016,Han2018,Han2019,YSun2019PushingCaching,BLv2019Cache,Lv2020TCOM,Lv2020JCIN}, which are designed for fixed users, cannot be applied to address the curse of dimensionality. Thus, a novel solution framework is proposed in this paper. Particularly, we first propose a baseline scheduling policy, whose value function can be derived analytically. Then, one-step policy iteration is applied based on the value function of the baseline policy to obtain the proposed sub-optimal policy.

	\item \textbf{An efficient reinforcement learning algorithm for system optimization without task arrival statistics:} The value function of the baseline policy depends on the task arrival statistics which may not be known in practice. Thus, we design a novel reinforcement learning method for evaluating the value function. The conventional reinforcement learning method, i.e. Q-learning, needs to learn the Q-function for all state-action pairs, which is infeasible in our problem due to the tremendous state and action spaces. In the proposed reinforcement learning method, by exploiting the derived expression of value function, we only need to track some statistical parameters. The learning efficiency is significantly improved. Moreover, we also design a stochastic gradient descent (SGD) algorithm to optimize the transmit power of the baseline policy without system statistics in an online manner, such that the performance of the proposed policy can be further improved.
	\item \textbf{Analytical performance bound:}  
In most of the existing approximate MDP methods, it is difficult to investigate the performance of the proposed algorithm analytically. In our proposed solution framework, we manage to obtain an analytical cost upper bound on the proposed algorithm.
\end{itemize}

The remainder of this paper is organized as follows. In Section II, the MEC system model is introduced. The MDP problem formulation is elaborated in Section III, and the approximate-MDP-based low-complexity scheduling framework is illustrated in Section IV. A novel reinforcement learning algorithm and a SGD-based optimization algorithm are designed in section V. Simulation results are shown in Section VI, and the conclusion is drawn in Section VII.

We use the following notation throughout this paper. $[X]^+$ denotes $\max\{X, 0\}$. $ \lceil X\rceil $ is the minimum integer greater than or equal to $ X $, and $ \lfloor X \rfloor $ is the maximum integer less than or equal to $ X $. $I(\cdot)$ is the indicator function. Bold uppercase $\mathbf A$ denotes a matrix or a system state. Bold lowercase $\mathbf a$ denotes a vector. $\mathbf A^{-1}$ and $\mathbf A^{\mathsf T}$ are the inverse and transpose of the matrix $\mathbf A$, respectively. $\mathbf I$ denotes the identity matrix with dimensionality implied by context. Calligraphic letter $\mathcal A$ denotes a set. $|\mathcal A|$ is the cardinality of $\mathcal A$, operator / denotes the set subtraction, and $\emptyset$ denotes the empty set.

\section{System Model}

\subsection{Network Model}
\begin{figure}[tb]
	\centering
	\includegraphics[scale=0.7]{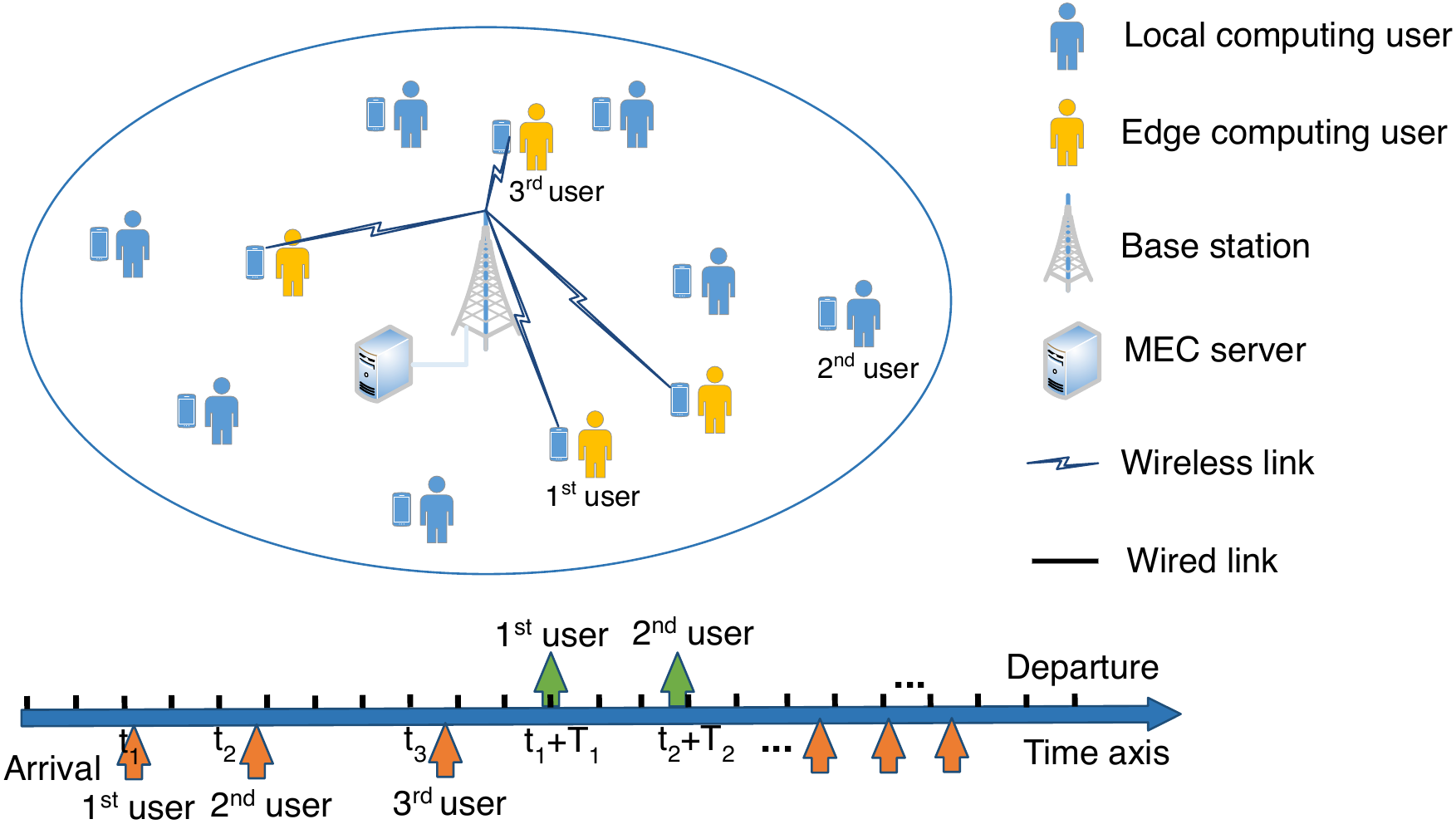}
	\caption{Illustration of MEC system model, where active devices arrive randomly in the cell coverage area and the time axis, and  the set of active devices in each frame is variable.}
	\label{fig: SystemModel}
\end{figure}

We consider a single-cell MEC system as illustrated in Fig. \ref{fig: SystemModel}, where a BS serves a region $ \mathcal C $ and an MEC server is connected to the BS. Mobile devices with computation tasks arrive randomly in the service region $ \mathcal C $. Binary computation offloading model is adopted, and every task is assumed to be indivisible from the computation perspective. Thus, each task can be either computed locally or offloaded to the MEC server via uplink transmission.

 The mobile devices with computation tasks are named as active devices. As illustrated in Fig. \ref{fig: SystemModel}, the time axis of computation and uplink transmission scheduling is measured in frame, each with duration $T_s$. Similar to \cite{KHan2018SpatialModelingMEC}, in each frame, there is at most one new active device arrived in the cell with probability $ P_N\in(0,1] $\footnote{Since we consider the scheduling in a single cell and the time scale of one frame is short (around 10ms), the average number of arrivals in one frame should be small for a reasonable burden of potential mobile edge computing. Use the Poisson arrival as an example, if the average number of arrivals in a frame is significantly smaller than 1, the probability that the number of arrivals is greater than 1 is negligible.}. The distribution density of the new active device is represented as $ \lambda ({\mathbf l}) $ for arbitrary location in the cell region $ {\mathbf l} \in \mathcal C $. Thus, 
$$\int_{\mathcal C} \lambda ({\mathbf l}) ds({\mathbf l}) = 1,$$
and
\begin{equation}\label{eq:lambda}
\Pr [\mbox{New \!active\! device\! in\! region } \mathcal C^{'}] \!\!= \!\!\!\int_{\mathcal C^{'}} \!\!\lambda ({\mathbf l}) ds({\mathbf l}),\!\! \ \forall \mathcal C^{'} \!\!\subseteq \mathcal \!C.
\end{equation}
Moreover, it is assumed that the location of each active device is quasi-static in the cell when its task is being transmitted to the MEC server. The active devices become inactive when their computation tasks are completed either locally or remotely at the MEC server, which is referred to as the departure of active devices. As in many of the existing works \cite{You2015SingleUserWPT,mao2016dynamicmec,zhang2013mccstochastic,Ji2019WPMEC}, it is assumed  that there are sufficiently many high-performance CPUs at the MEC server so that the computing latency at the MEC server can be neglected compared with the latency of local computing or uplink transmission. Moreover, due to relatively smaller sizes of computation results, the downloading latency of computation results is also neglected as in \cite{mao2016power-delay,mao2016dynamicmec,zhang2013mccstochastic,Ji2019WPMEC}.

Every new active device in the cell is assigned with a unique index. Let $ \mathcal U_L (t) $ and $ \mathcal U_E (t) $ be the sets of active devices in the $ t $-th frame whose tasks are computed locally and at the MEC server respectively, $ \mathcal{D}_{L}(t) \subseteq \mathcal{U}_L(t) $ and $\mathcal{D}_{E}(t) \subseteq \mathcal{U}_{E}(t)$ be the subsets of active devices whose computation tasks are accomplished in the $t$-th frame locally and at the MEC server respectively, $ n_t $ be the index of the new active device arriving at the beginning of $ t $-th frame.  If there is no active device arrival at the beginning of $ t $-th frame,  $ \{n_t\} = \emptyset $. On the other hand, if there is a new active device arrival at the beginning of a frame, the BS should determine if the computation task is computed at the device or the MEC server. Let $ e_t \in \{0,1\} $ represents the decision, where $ e_t = 1 $ means the task is offloaded to the MEC server and $ e_t=0 $ means otherwise. The dynamics of active devices can be represented as
\begin{equation}
\mathcal{U}_E(t+1) = \left\{ \begin{array}{cc}
 \mathcal{U}_E(t) \cup \{n_{t}\} /\mathcal{D}_E(t)&  \mbox{when } e_t = 1, \\
 \mathcal{U}_E(t)/\mathcal{D}_E(t)&  \mbox{otherwise,}
\end{array}   \right.
\end{equation}
\begin{equation}
\mathcal{U}_L(t+1) = \left\{ \begin{array}{cc}
\mathcal{U}_L(t) \cup \{n_{t}\} /\mathcal{D}_L(t)&  \mbox{when } e_t= 0, \\
\mathcal{U}_L(t)/\mathcal{D}_L(t)&  \mbox{otherwise.}
\end{array}   \right.
\end{equation}

\subsection{Task Offloading Model}

The input data for each computation task is organized by segments, each with $ b_s $ bits. Let $ d_k $ be the number of input segments for the task of the $ k $-th active device. It is assumed that the number of segments for each task (say the $k$-th active device) is a uniformly distributed random integer between $d_{\min}$ and $d_{\max}$ whose probability mass function (PMF) is given by\footnote{As a remark, notice that our proposed solution can be trivially extended to other distributions of task size.}
\begin{align} \label{eq:pd}
\Pr[d_k=a]=\left\{ \begin{array}{cc}
	\frac{1}{d_{\max}-d_{\min}+1} &\ \mbox{for}\ d_{\min}\leq a\leq d_{\max},\\
	0  &\mbox{otherwise.}
   \end{array}   \right.
\end{align}
 For the computation tasks to be offloaded to the MEC server, the input data should be delivered to the BS via uplink transmission. Hence, an uplink transmission queue is established at each active device for edge computing. Let $ Q^E_k(t) $, $ \forall k \in \mathcal U_E(t) $, be the number of segments in the uplink transmission queue of the $ k $-th active device at the beginning of the $ t $-th frame. Hence, for all $t$ with $\{n_t\} \neq \emptyset$ and $e_t=1$,
\begin{equation*}
Q_{n_t}^E(t+1) = d_{n_t}.
\end{equation*}

In the uplink, it is assumed that only one active device is selected in one uplink frame and the uplink transmission bandwidth is denoted as $W$.  Let $$ H_k(t)=\sqrt{\rho_k}h_k(t) ,\forall k\in \mathcal U_E(t) $$ be the uplink channel state information (CSI) from the $ k $-th active device to the BS, where $ h_k(t) $ and $ \rho_k  $ represent the small-scale fading and pathloss coefficients respectively. $ h_k(t) \sim \mathbb {CN} (0,1) $ is complex Gaussian distributed with zero mean and variance $ 1 $. Moreover, it is assumed that $ h_k(t) $ is independently and identically distributed (i.i.d.) for different $ t $ \footnote{The small-scale fading is varying due to the motion of transmitter, receiver or the scatters. Moreover, as described in Section 3.3.3 of \cite{goldsmith_2005}, the small-scale fading coefficients can be treated as independent as long as the frame duration is larger than the channel coherent time.} and $ k $. Let $ p_k(t) $ be the uplink transmission power of the $ k $-th active device if it is selected in the $ t $-th frame. The uplink channel capacity of the $ k $-th active device, if it is selected in the $ t $-th frame, can be represented by
\begin{equation*}
r_k(t) = W \log_2\left(1+\frac{p_k(t) \rho_k |h_k(t)|^2}{\sigma_z^2}\right),
\end{equation*}
where $\sigma_z^2$ is the power of white Gaussian noise. Furthermore, the number of segments transmitted within the $t$-th frame can be obtained by
\begin{equation}
	\phi_k(t)=\biggl\lfloor \frac{r_k(t)T_s}{b_s} \biggr\rfloor.
\end{equation}
 Hence, let $ a_t $ be the index of the selected uplink transmission device in the $ t $-th frame, we have the following queue dynamics for all $k\in \mathcal U_E(t)$,
\begin{equation}
	Q^E_k(t+1)=\left\{ \begin{array}{cc}
		\left[ Q^E_k(t)-\phi_k(t) \right]^+ & \mbox{if } k=a_t,\\
		Q^E_k(t) & \mbox{if } k\neq a_t.
	\end{array} \right.
\end{equation}

	Moreover, the $k$-th active device become inactive in the $(t+1)$-th frame ($k \in \mathcal{D}_E(t)$), if $Q^{E}_{k}(t) \neq 0$ and $Q^{E}_{k}(t+1) =0$.

\begin{Remark}[Variable Uplink Queue Number]
	In the existing works considering resource allocation with multiple queues, such as \cite{YCui2010AMDP-OFDM-downlink,YCui2010AMDP-OFDM-uplink,RWang2011DistTowHopMIMO, RWang2013RelayApproxMDP,Han2016,Han2018,Han2019,YSun2019PushingCaching,BLv2019Cache,Lv2020TCOM,Lv2020JCIN}, there are fixed active devices in the cell. In this paper, the number of active devices is variable. Hence, the queue state (number of data segments in all the queues) in the existing works can be represented by a vector with fixed dimension, but the queue state in this paper has to be represented by a vector with variable dimension. This will raise challenge in the approximate-MDP-based scheduler design, as the existing approaches adopted in \cite{YCui2010AMDP-OFDM-downlink,YCui2010AMDP-OFDM-uplink,RWang2011DistTowHopMIMO,
		 RWang2013RelayApproxMDP,Han2016,Han2018,Han2019,YSun2019PushingCaching,BLv2019Cache,Lv2020TCOM,Lv2020JCIN} cannot be applied with variable queue number. As elaborated later, we shall propose a novel approximate MDP framework to address this issue.
\end{Remark}

\subsection{Local Computing Model}
Following the computation models in \cite{you2016multiuser,mao2016power-delay}, the average number of CPU cycles for computing one bit of the input task data of the $k$-th active device is denoted as $\ell_k$, which is determined by the types of applications. Denote the local CPU frequency of the $k$-th active device as $f_k$ which is assumed to be a constant for each device and may vary over devices. We assume $\ell_k$ and $f_k$ are both random variables whose probability density functions (PDFs) are denoted by $\pi_{\ell}$ and $\pi_f$ respectively. Thus, 
\begin{equation}\label{eq:pl}
\Pr[\ell_1 \leq \ell < \ell_2]=\int_{\ell_1}^{\ell_2}\pi_{\ell}(\ell)d\ell,
\end{equation}
and 
\begin{equation}\label{eq:pf}
\Pr[f_1\leq f< f_2]=\int_{f_1}^{f_2}\pi_f(f)df.
\end{equation}
 
An input data queue is established at each local computing device. Let $ Q^L_k(t) $, $ \forall k \in \mathcal U_L(t) $, be the number of segments in the input data queue of the $ k $-th active device for local computing at the beginning of the $ t $-th frame. Hence, for all $t \mbox{ with } \{n_t\} \neq \emptyset \mbox{ and } e_t= 0$,
\begin{equation*}
Q_{n_t}^L(t+1) = d_{n_t}.
\end{equation*}
The queue dynamics at all active local computing devices can be written as
\begin{equation}
Q^L_k(t+1)= \left[Q^L_k(t)-\frac{f_k T_s}{\ell_k b_s}\right]^+, \ \forall k \in \mathcal U_L(t).
\end{equation}

Hence, the $k$-th active device is inactive in the $(t+1)$-th frame ($k \in \mathcal{D}_L(t)$), if $Q^{L}_{k}(t) \neq 0$ and $Q^{L}_{k}(t+1) =0$.
Moreover, the total computation time (measured by frames) for $k$-th active device, whose task is computed locally, is given by
\begin{equation}
T_{loc}(d_k,f_k,\ell_k)=\biggl\lceil\frac{d_k b_s \ell_k}{f_k T_s}\biggr\rceil.
\end{equation}
Following the power consumption model in  \cite{Burd1996Processor}, the local computation power of $k$-th device is
\begin{equation}
p_{loc}(f_k)=\kappa f_k^3,
\end{equation}
where $\kappa$ is the effective switched capacitance related to the CPU architecture.

\section{Problem Formulation}
In this section, we formulate the optimization of task offloading decision, uplink device selection and power allocation as an infinite-horizon MDP problem with discounted cost, where the random active device arrivals are taken into consideration.

\subsection{System State and Scheduling Policy}

The system state and scheduling policy are defined as follows.

\begin{Definition}[System State]
	At the beginning of $ t $-th frame, the state of the MEC system is uniquely specified by $ \mathbf{S}_t =(\mathbf{S}^E_t,\mathbf{S}^L_t,\mathbf{S}^N_t)$, where
	\begin{itemize}
		\item $\mathbf S^E_t$ specifies the status of task offloading, including the set of edge computing devices $ \mathcal{U}_E(t) $, their uplink small-scale fading coefficients $ \mathcal H_E(t)\triangleq \{h_k(t)|k\in\mathcal{U}_E(t)\}$, pathloss coefficients $ \mathcal{G}_E(t)\triangleq \{  \rho_k | k \in \mathcal{U}_E(t) \} $, and their uplink queue state information (QSI) $ \mathcal{Q}_E(t) \triangleq \{Q^E_k(t)| k\in \mathcal{U}_E(t)\}$.

		\item $\mathbf S^L_t$ specifies the status of local computing, including the set of local computing devices $\mathcal U_L(t)$, the application-dependent parameters $\mathcal {L}(t)\triangleq \{\ell_k| k\in \mathcal U_L(t)\}$, their CPU frequencies $\mathcal F(t)\triangleq \{f_k|k\in \mathcal U_L(t)\}$, and their QSI $\mathcal{Q}_L(t)\triangleq \{Q^L_k(t)|k\in \mathcal U_L(t)\}$.

		\item $\mathbf S^N_t$ specifies the status of the new active device, including the indicator of new arrival $I_N(t)\triangleq I(\{n_t\} \neq \emptyset)$, its index $ n_t $, pathloss coefficient $ \rho_{n_t}(t) $, size of input data $ d_{n_t} ,$ CPU frequency $f_{n_t}$ and $\ell_{n_t}$.
	\end{itemize}

\end{Definition}

\begin{Definition}[Scheduling Policy]
The scheduling policy $ \Omega(\mathbf{S}_t)  \triangleq \left(a_t, p(t), e_t\right)$ is a mapping from the system state $ \mathbf{S}_t $ to the scheduling actions, i.e, the index $a_t$ of the selected uplink transmission device in the $t$-th frame and its transmission power $p(t)$, as well as  the offloading decision $e_t$ for the new arriving active device (if any).
\end{Definition}

\begin{Remark}[Huge Space of System State]\label{Rem:HugeSpace}
	Since the arrival and departure of active devices are considered, the space of system state, denoted as $\mathcal S$, is more complicated than the existing works with fixed users. Take $\mathbf S^E_t$ as the example. The cardinality of $\mathcal U_E(t)$ is not a constant, hence $\mathcal U_E(t)$ with all possible cardinalities should be included in the system state space. Moreover, given a $\mathcal U_E(t)$, all possible small-scale fading and pathloss coefficients, $\mathcal H_E(t)$ and $\mathcal G_E(t)$, should also be included in the system state space. So does the QSI. Note that the spaces of small-scale fading and pathloss coefficients are continuous. In this paper, we shall address the low-complexity algorithm design with such huge system state space.
\end{Remark}

\subsection{Problem Formulation of MEC Scheduling}

According to Little's law, the average latency of a task is proportional to the average number of active devices in the system\cite{Kleinrock1975QueueSystem}. Hence, we define the following weighted sum of the number of active devices and their power consumptions as the system cost in the $ t $-th frame.
\begin{equation*}
g(\mathbf{S}_t, \Omega(\mathbf{S}_t))\! \triangleq\! w(|\mathcal{U}_E(t)| + |\mathcal{U}_L(t)|)+ p(t) +\!\!\!\sum_{k\in\mathcal U_L (t)} p_{loc}(f_k),
\end{equation*}
where $w$ is the weight on the latency of mobile devices. The overall minimization objective with the initial system state $ \mathbf{S} $ is then given by
\begin{eqnarray}
\overline{G} (\Omega, \mathbf{S})\!\!\!\!\!\! &\triangleq& \!\!\!\!\!\! \lim\limits_{T \rightarrow + \infty }\!\!\mathbb{E}_{\{\mathbf{S}^N_t,\mathcal{H}_E(t)|\forall t\}}\bigg[\sum_{t=1}^{T} \gamma^{t-1} g(\mathbf{S}_t, \Omega(\mathbf{S}_t))\! \bigg| \mathbf{S}_1\!=\!\mathbf{S} \!\bigg], \nonumber
\end{eqnarray}
where $ \gamma $ is the discount factor. Thus, the MEC scheduling problem is formulated as the following infinite-horizon MDP.
\begin{Problem}[MEC Scheduling Problem]\label{prob:main}
\begin{eqnarray}
\Omega^{*}=\arg\min_{\Omega} \overline{G} (\Omega, \mathbf{S}).
\end{eqnarray}
\end{Problem}

According to \cite{Bertsekas2012Dynamic}, the optimal policy of Problem \ref{prob:main} can be obtained by solving the following Bellman's equations.
\begin{align}
V(\mathbf{S}_t)  =& \min\limits_{\Omega(\mathbf{S}_t)} \bigg[g\left(\mathbf{S}_t,\Omega(\mathbf{S}_t)\right) + \sum_{\mathbf{S}_{t+1}} \gamma \Pr\!\left(\mathbf{S}_{t+1} | \mathbf{S}_{t}, \Omega(\mathbf{S}_t)\right)V(\mathbf{S}_{t+1}) \bigg], \forall \mathbf{S}_t \in \mathcal S. 
\end{align}
where $ V(\mathbf{S}) $ is the value function for system state $ \mathbf{S} $. Generally speaking, standard value iteration can be used to solve the value function, and the optimal policy denoted as $\Omega^*$ can be derived by solving the minimization problem of the right-hand-side of the above Bellman's equations. In our problem, however, the conventional value iteration is intractable due to the large state space. Conventional value iteration should evaluate the value function for all system state in  the state space $\mathcal S$. However, as mentioned in Remark \ref{Rem:HugeSpace} that the spaces of small-scale fading and pathloss coefficients are continuous. Even the continuous spaces of  small-scale fading and pathloss coefficients can be quantized, the state space grows exponentially with respect to (w.r.t.) the number of active devices.

In order to address the above issues, similar to \cite{WANG2017,Ruiwang2018_GLOBECOM}, we first reduce the system state space by exploiting (1) the independent distributions of small-scale fading and new active devices' statistics in each frame, and (2) the deterministic cost (given system state) of local computing devices. Specifically, the optimal policy can also be derived via the following equivalent Bellman's equations w.r.t compact system states.

\begin{Lemma}[Bellman's Equations with Compact State Space] \label{lem: BE_reduced}
Define the local computing cost of the $n_t$-th active device as 
$$C(n_t) \triangleq \sum_{\tau=1}^{T_{loc}(d_{n_t},f_{n_t},\ell_{n_t})}\!\gamma^{\tau} \left[ w + p_{loc}(f_{n_t})\right]$$
 and the compact system state as 
 $$ {\widetilde{\mathbf{S}}_t} \triangleq (\mathcal U_E(t), \mathcal G_E(t), \mathcal Q_E(t)).$$
  Let
$$g'(\mathbf{S}_t,\Omega(\mathbf S_t))\triangleq w|\mathcal{U}_E(t)| + p(t) + I_N(t) (1-e_t) C(n_t),$$
\begin{eqnarray}
W(\widetilde {\mathbf S})\triangleq \min_{\Omega} \lim\limits_{T \rightarrow + \infty }\!\mathbb{E}_{\{\mathcal H_E(t), \mathbf S^N_t|\forall t\}} \bigg[\sum_{t=1}^{T} \gamma^{t-1} \!g'(\mathbf{S}_t, \Omega(\mathbf{S}_t)) \bigg| \mathbf{\widetilde S}_1=\mathbf{\widetilde S} \bigg]. \nonumber
\end{eqnarray}
They satisfy the following Bellman's equations.
\begin{eqnarray}
W(\widetilde{\mathbf{S}}_t) =\min\limits_{\Omega(\mathbf{S}_t)}\mathbb{E}_{\{\mathcal H_E(t), \mathbf S^N_t|\forall t\}} \bigg\{g'\left( \mathbf{S}_t,\Omega(\mathbf{S}_t)\right)+ \sum_{\widetilde{\mathbf{S}}_{t+1}}\gamma \Pr\left(\widetilde{\mathbf{S}}_{t+1} | \mathbf{S}_{t}, \Omega(\mathbf{S}_t)\right)W (\widetilde{\mathbf{S}}_{t+1}) \!\bigg\}, \forall \widetilde{\mathbf S}_t.
\label{eqn:reduced-problem}
\end{eqnarray}
Moreover, the scheduling policy minimizing the right-hand-side of the above equation is the optimal policy of Problem \ref{prob:main}.
\end{Lemma}

\begin{proof}
	Please refer to appendix A.
\end{proof}

	In Lemma \ref{lem: BE_reduced}, the new value function $W(\widetilde{\mathbf{S}}_t)$ depends only on the compact system state $\widetilde {\mathbf S}_t$. Although the state space of the MDP problem is significantly reduced, it is still infeasible to solve equation (\ref{eqn:reduced-problem}) via the conventional value iteration. This is because the space of $\widetilde{\mathbf S}_t$ is still huge, as mentioned in Remark \ref{Rem:HugeSpace}. Moreover, the conventional approximate MDP method introduced in \cite{YCui2010AMDP-OFDM-downlink,YCui2010AMDP-OFDM-uplink,RWang2011DistTowHopMIMO, RWang2013RelayApproxMDP,YSun2019PushingCaching,BLv2019Cache}, e.g., via parametric approximation architecture, requires a fixed number of quasi-static mobile devices. It cannot be applied to our problem. In the following section, we shall propose a novel low-complexity solution framework, which approximates $W(\widetilde {\mathbf S})$ with analytical expression and obtains a sub-optimal policy by minimizing the right-hand-side of (\ref{eqn:reduced-problem}).

\section{Low-Complexity Scheduling}

In order to obtain a low-complexity scheduling policy, we first introduce a heuristic scheduling policy as the baseline policy in Section IV-A, whose value function are derived analytically. Then in Section IV-B, the proposed low-complexity sub-optimal policy can be obtained via the above value function and one-step policy iteration. 
It can be proved that
the derived value function of the baseline policy is the cost upper bound of the proposed sub-optimal policy.

\subsection{Baseline Scheduling Policy}
The following policy is adopted as the baseline scheduling policy.

\begin{Policy}[Baseline Scheduling Policy $\Pi$] Given the system state $\mathbf{S}_t$ of the $t$-th frame ($\forall t$), the baseline scheduling policy $\Pi(\mathbf{S}_t)=\left(a_t, p(t), e_t\right)$ is provided below.
\begin{itemize}
\item Uplink device selection $a_t = \min \mathcal{U}_E(t)$, $\forall t$. Thus, the BS schedules the uplink device in a first-come-first-serve manner.
\item The transmission power $p(t)$ compensates the large-scale fading (link compensation). Thus,
\begin{align}
p(t)=\frac{p_r}{\rho_{a_t}}, \forall t,
\end{align}
where $p_r$ is the average receiving power at the BS.
\item  The task of the new active device is offloaded to MEC server only when there are less than $K$ active edge computing devices in the system, i.e.,
	\begin{align}
	e_t=I(\mathcal{|U}_E(t)|<K), \forall t.
	\end{align}
\end{itemize}
\end{Policy}

\begin{figure}[tb]
	\centering
	\includegraphics[scale=0.5]{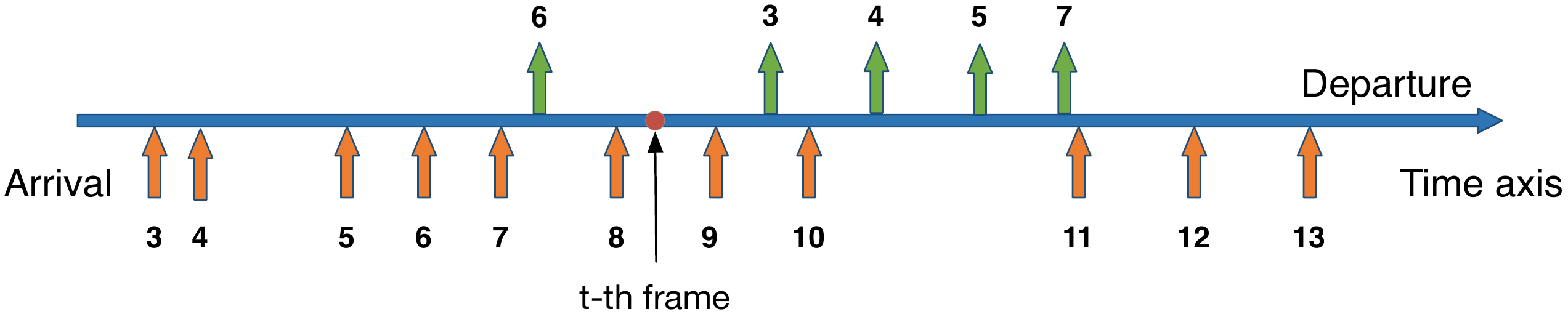}
	\caption{An example to illustrate the baseline policy.}
	\label{fig:timeline}
\end{figure}

\begin{Example}
An example illustrating the baseline policy is described below. Suppose $\mathcal U_E(t)=\{3, 4, 7\}$, $\mathcal U_L(t)=\{5, 8\}$ in the $t$-th frame, as illustrated in Fig. \ref{fig:timeline}. If the baseline policy $\Pi$ is used since the (t+1)-th frame with $K=2$, the 3-rd active device will transmit first, followed by the 4-th and 7-th active devices. Their transmission powers are $\frac{p_r}{\rho_3}$, $\frac{p_r}{\rho_4}$ and $\frac{p_r}{\rho_7}$, respectively. 
 Before the completeness of task offloading of the $3$-rd and $4$-th active devices, all new active devices will be scheduled for local computing. For example, the $9$-th and $10$-th active devices are scheduled for local computing. Then, the $11$-th and $12$-th active devices are scheduled for edge computing and the $13$-th one is scheduled for local computing.
\end{Example}
Given the initial compact system state in the first frame $\widetilde{\mathbf{S}}$, the value function of policy $\Pi$, measuring the cost of the
baseline policy since the first frame, is defined as
\begin{align}
W_{\Pi}(\widetilde{\mathbf{S}})\triangleq \lim\limits_{T\to+\infty} \mathbb{E}_{\{\mathbf{S}_t|\forall t\}}^{\Pi} \Big[ \sum_{t=1}^{T}\gamma^{t-1}{g}'(\mathbf{S}_t,\Pi(\mathbf{S}_t)) \Big|\widetilde{\mathbf{S}}_1=\widetilde{\mathbf{S}}
\Big
].
\end{align} 

In order to derive the analytical expression of $W_{\Pi}(\widetilde{\mathbf{S}})$, we denote the index of the $ k $-th active device in $ \mathcal{U}_E (1)$  as $ m_{k} $, i.e. $\mathcal U_E(1)=\{m_1, m_2, \dots, m_{|\mathcal U_E(1)|}\}$, and $ T_{k} $ as the number of frames for completing the uplink transmission of the $ m_{k} $-th device. The calculation of $W_{\Pi}(\widetilde {\mathbf S})$ in infinite time horizon can be decomposed into three periods: (1) the transmission period of the first $\left[|\mathcal{U}_E(1)|-K\right]^{+}$ active devices in the edge computing device set $\mathcal{U}_E(1)$; (2) the transmission period of the last $\min(K,|\mathcal{U}_E(1)|)$ active devices in the edge computing device set $\mathcal{U}_E(1)$; (3) the remaining frames to infinity. The costs of these three periods, denoted as $W_{\Pi}^{(1)}(\widetilde {\mathbf S})$, $W_{\Pi}^{(2)}(\widetilde {\mathbf S})$ and $W_{\Pi}^{(3)}(\widetilde {\mathbf S})$, are defined in the followings.
\begin{align}
&W_{\Pi}^{(1)}(\widetilde{\mathbf{S}})\triangleq  \mathbb{E}_{\{\mathbf{S}_t|\forall t\},\{T_{k}|\forall k\}}^{\Pi} \left[ 
\sum_{t=1}^{\sum_{i=1}^{\left[|\mathcal{U}_E(1)|-K\right]^{+}}T_{i}}\gamma^{t-1}{g}'(\mathbf{S}_t,\Pi(\mathbf{S}_t)) \Bigg|\widetilde{\mathbf{S}}_1=\widetilde{\mathbf{S}}
\right],\label{eqn:W1}\\
&W_{\Pi}^{(2)}(\widetilde{\mathbf{S}})\triangleq  \mathbb{E}_{\{\mathbf{S}_t|\forall t\},\{T_{k}|\forall k\}}^{\Pi} \left[ 
\sum_{t=\sum_{i=1}^{\left[|\mathcal{U}_E(1)|-K\right]^{+}}T_{i}+1}^{\sum_{i=1}^{|\mathcal{U}_E(1)|}T_{i}}\gamma^{t-1}{g}'(\mathbf{S}_t,\Pi(\mathbf{S}_t)) \Bigg|\widetilde{\mathbf{S}}_1=\widetilde{\mathbf{S}}
\right],\label{eqn:W2}\\
&W_{\Pi}^{(3)}(\widetilde{\mathbf{S}})\triangleq \lim\limits_{T\to +\infty} \mathbb{E}_{\{\mathbf{S}_t|\forall t\},\{T_{k}|\forall k\}}^{\Pi} \left[ 
\sum_{t=\sum_{i=1}^{|\mathcal{U}_E(1)|}T_{i}+1}^{T}
\gamma^{t-1}{g}'(\mathbf{S}_t,\Pi(\mathbf{S}_t)) \Bigg|\widetilde{\mathbf{S}}_1=\widetilde{\mathbf{S}}
\right].\label{eqn:W3}
\end{align}

Hence,
\begin{align}
W_{\Pi}(\widetilde {\mathbf S})=W_{\Pi}^{(1)}(\widetilde {\mathbf S})+W_{\Pi}^{(2)}(\widetilde {\mathbf S})+W_{\Pi}^{(3)}(\widetilde {\mathbf S}).
\end{align}

The per-frame system cost in the three periods can be treated as stochastic processes, and the discounted summation of per-frame cost can be calculated via the probability  transition matrices. Specifically, 
the expressions $W_{\Pi}^{(1)}(\widetilde {\mathbf S})$, $W_{\Pi}^{(2)}(\widetilde {\mathbf S})$ and $W_{\Pi}^{(3)}(\widetilde {\mathbf S})$ are given by following
three lemmas respectively.

\begin{Lemma}[Analytical Expression of $W_{\Pi}^{(1)}(\mathbf {\widetilde S})$]\label{lem:W1}
	$W_{\Pi}^{(1)}(\widetilde{\mathbf{S}})$ can be written by 
	\begin{align}\label{eq:W_Pi_1}
	W_{\Pi}^{(1)}(\widetilde{\mathbf{S}}) = & \mathbb{E}_{\{T_{k}|\forall k\}}\left [ \sum_{ k =1}^{\left[|\mathcal{U}_E(1)|-K\right]^{+}}\gamma^{\sum\limits_{i=1}^{k-1}T_{i}} \left(  \frac{1-\gamma^{T_{k}}}{1-\gamma} \frac{p_r}{\rho_{m_{k}}} + w\Big[|\mathcal{U}_E(1)|-k+1\Big]^{+}\frac{1-\gamma^{\sum\limits_{i=1}^{k}T_{i}}}{1-\gamma}\right)\right ]\nonumber\\
	&+P_N \mathbb{E}_{\{T_k|\forall k\}}\left[\sum_{t=1}^{\sum_{k=1}^{\left[|\mathcal{U}_E(1)|-K\right]^{+}}T_{k}} \gamma^{t-1}\mathbb{E}[C(n_t)] \right],
	\end{align}
where 
\begin{align}
\mathbb{E} [C(n_t)]=\sum_{d_{\min}}^{d_{\max}}\frac{\int\int \pi_f(f)\pi_{\ell}(\ell)C(n_t)dfd\ell}{d_{\max}-d_{\min}+1}. 
\end{align}
Moreover, for sufficiently large input data size, we have
\begin{align}\label{eq:Tk}
T_{k} = \left\lceil\frac{ Q_{m_{k}} b_s }{\mathbb{E}_{
	h} W\log_2 \left( 1 + \frac{p_r |h|^2}{\sigma_z^2} \right)T_s} \right\rceil, \forall k,
\end{align}
where $ \mathbb{E}_{h}  $ is the expectation w.r.t. small-scale fading.
\end{Lemma}
\begin{proof}
Please refer to Appendix B.
\end{proof}

\begin{Lemma}[Analytical Expression of $W_{\Pi}^{(2)}(\mathbf {\widetilde S})$]\label{lem:W2}
	Define the following notations:
	\begin{itemize}
		\item $\mathbf u \in \mathbb R^{(K+1)\times 1}$, whose $(\min(|\mathcal{U}_E(1)|,K)+1)$-th entry is $1$ and
		other entries are all $0$. 
		\item $ \mathbf{g} = [g_1 \ g_2 \ ... \ g_{K+1}]^{\mathsf T} \in \mathbb R^{(K+1)\times 1}$, where $g_1=0$, $g_i=w(i-1)$,  $\forall i=2,3,...,K$, and 
		$g_{K+1}=wK+P_N\mathbb{E}[C(n_t)]$.
		\item $\mathbf{P}\in  \mathbb R^{(K+1)\times(K+1) } $, where $[\mathbf{P}]_{i,i-1}=1$, $\forall i=2,3,...,K+1$, $[\mathbf{P}]_{i,i}=1$ and
		other entries are all $0$.
		\item  $\mathbf{M}\in  \mathbb R^{(K+1)\times(K+1) } $, where $[\mathbf{M}]_{j,j+1}=P_N$, $[\mathbf{M}]_{j,j}=1-P_N$, $\forall j=1,2,...,K$, $[\mathbf{M}]_{K+1,K+1}=1$,  and
		other entries are all $0$.
	\end{itemize}
	Then, the analytical expression of 	$W_{\Pi}^{(2)}(\widetilde{\mathbf{S}})$ is given by
	\begin{align}\label{eqn:W_2}
	W_{\Pi}^{(2)}(\widetilde{\mathbf{S}}) =&  \mathbb{E}_{\{T_{k}|\forall k\}}\left [ 
	\sum_{k=\left[|\mathcal{U}_E(1)|-K\right]^{+}+1}^{|\mathcal{U}_E(1)|}
	\gamma^{\sum\limits_{i=1}^{k-1}T_{i}} \left(  \frac{1-\gamma^{T_{k}}}{1-\gamma} \frac{p_r}{\rho_{m_{k}}}\right)\right ]\nonumber\\
	&+	\sum_{k=\left[|\mathcal{U}_E(1)|-K\right]^{+}+1}^{|\mathcal{U}_E(1)|} \ \sum_{t=\sum_{i=1}^{k-1}T_{i}+1}^{\sum_{i=1}^{k}T_{i}} \gamma^{t-1}\mathbf{u}_{k}^{\mathsf T}(\mathbf{M})^{\beta_{k,t}}\mathbf{g},
	\end{align}
	where $\beta_{k,t} \triangleq t-\sum_{i=1}^{k-1}T_{i}-1$, and $\mathbf{u}_{\left[|\mathcal{U}_E(1)|-K\right]^{+}+1}=\mathbf{u},$
	\begin{align*}
	\mathbf{u}_k=\left[ \mathbf{u}_{k-1}^{\mathsf T}(\mathbf{M})^{T_{k-1}}\mathbf{P}\right]^{\mathsf T}, \ k = \left[|\mathcal{U}_E(1)|-K\right]^{+}+2, \dots, |\mathcal{U}_E(1)|.
	\end{align*}

\end{Lemma}

\begin{proof}
Please refer to Appendix B.
\end{proof}

\begin{Lemma}[Analytical Expression of $W_{\Pi}^{(3)}(\mathbf {\widetilde S})$]\label{lem:W3}
	Define the following notations:
	\begin{itemize}
		\item $\zeta\in\{0,1,\dots,K\}$ denotes the number of edge computing devices.
		\item $\xi\in \{0,1,\dots,d_{\max}\}$ denotes the number of segments of the first edge computing device.
		\item $\epsilon_{\zeta,\xi}$ denotes an index and
		\begin{equation}
		\epsilon_{\zeta,\xi} \triangleq \left\{ \begin{array}{cc}
		1&  \zeta=0,\\
		(\zeta-1)d_{\max}+\xi+1	&  \mbox{otherwise.}
		\end{array}   \right.
		\end{equation}
		\item $\mathbf{v}\in \mathbb R^{(Kd_{\max}+1)\times 1}$. When $|\mathcal{U}_E(1)|=0$, $\mathbf{v}=[1 \ 0 \dots \ 0 \ 0]^{\mathsf T}$; otherwise,
		the entries of $\mathbf{v}$ is given by 
			\begin{equation} \label{eqn:v}
		[\mathbf{v}]_{	\epsilon_{\zeta,\xi} }\triangleq \left\{ \begin{array}{cc}
		\left[\mathbf{u}_{|\mathcal{U}_E(1)|}^{\mathsf T}(\mathbf{M})^{T_{|\mathcal{U}_E(1)|}}\mathbf{P}\right]_{1}		&  \epsilon_{\zeta,\xi}=1, \\
		\frac{1}{d_{\max}-d_{\min}+1}\left[\mathbf{u}_{|\mathcal{U}_E(1)|}^{\mathsf T}(\mathbf{M})^{T_{|\mathcal{U}_E(1)|}}\mathbf{P}\right]_{i}		&  \zeta=i-1, \ d_{\min}\leq\xi\leq d_{\max
		},
		\\
		0	&  \mbox{otherwise.}
		\end{array}   \right ..
		\end{equation}

		\item $\mathbf{c}\in \mathbb R^{(Kd_{\max}+1)\times 1}$, and 
		\begin{equation}
		[\mathbf{c}]_{	\epsilon_{\zeta,\xi} }\triangleq \left\{ \begin{array}{cc}
		0		&  \epsilon_{\zeta,\xi}=1, \\
		w\zeta+\mathbb{E}_{\rho_{n_t}}[\frac{p_r}{\rho_{n_t}}]	&  0<\zeta<K,
		\\
		w\zeta+\mathbb{E}_{\rho_{n_t}}[\frac{p_r}{\rho_{n_t}}]+P_N\mathbb{E}[C(n_t)]		&  \zeta=K.
		\end{array}   \right..
		\end{equation}
	\end{itemize}
	Then, the analytical expression of 	$W_{\Pi}^{(3)}(\widetilde{\mathbf{S}})$ is given by
		\begin{align}\label{eq:W-Pi-3}
	W_{\Pi}^{(3)}(\widetilde{\mathbf{S}})=\lim\limits_{T\to +\infty}\sum_{t=\sum_{i=1}^{|\mathcal{U}_E(1)|}T_{i}+1}^{T} \gamma^{t-1} \mathbf{v}^{\mathsf T}({\bf \Phi })^{t-\sum_{i=1}^{|\mathcal{U}_E(1)|}T_{i}-1}
	\mathbf{c}=\gamma^{\sum_{i=1}^{|\mathcal{U}_E(1)|}T_{i}}\mathbf{v}^{\mathsf T}\left( \mathbf{I}-\gamma{\bf \Phi} \right)^{-1}\mathbf{c}
	,
	\end{align}
 where
	the non-zero entries of the transition probability matrix ${\bf \Phi}\in \mathbb{R}^{(Kd_{\max}+1)\times (Kd_{\max}+1)}$ are given in table \ref{table:Phi}, and other entries are all $0$.
\end{Lemma}

\begin{proof}
	Please refer to Appendix B.
\end{proof}

Compared with optimal MDP solution\cite{Bertsekas2012Dynamic} and conventional approximate MDP\cite{YCui2010AMDP-OFDM-downlink,YCui2010AMDP-OFDM-uplink,RWang2011DistTowHopMIMO, RWang2013RelayApproxMDP,YSun2019PushingCaching,BLv2019Cache}, our proposed method can significantly reduce the complexity in the phase of value iteration. Particularly, the complexity of value function calculation for an arbitrary system state is  $\mathcal{O}(1)$, since we
can obtain the analytical expression of the approximate value function.

\begin{table*}
	\centering  
	\caption{Non-zeros Entries of Matrix ${\bf \Phi}$ ( $\alpha(x)=[2^{\frac{xb_s}{WT_s}}-1]\sigma_z^2$)}  
	\label{table:Phi}  
	\begin{tabular}{|c|c|c|c|c|}  
		\hline  
		& & & & \\[-6pt]  
		$\zeta$&$\xi$&${\zeta}'$ & ${\xi}'$ & $[{\bf \Phi}]_{\epsilon_{\zeta,\xi},\ \epsilon_{{\zeta}',{\xi}'}}$ \\  
		\hline
		0 & 0& 0 &0 & $1-P_N$
		\\
		\hline
		0 & 0 & 1& $d_{\min},\dots,d_{\max}$ & $\frac{P_N}{d_{\max}-d_{\min}+1}$ \\
		\hline
		1 & $1,\dots,d_{\max}$ & 0& 0& $(1-P_N)\exp\{-\frac{\alpha(\xi)}{p_r}\}$\\
		\hline
		$2,\dots,K-1$ & $1,\dots,d_{\max}$  & $\zeta-1$ & $d_{\min}, \dots, d_{\max}$ & $\frac{1-P_N}{d_{\max}-d_{\min}+1}\exp\{-\frac{\alpha(\xi)}{p_r}\}$
		\\
		\hline
		$1,\dots,K-1$ & $1,\dots,d_{\max}$  & $\zeta+1$ & $1,...,\xi$ & $P_N\left(\exp\{-\frac{\alpha(\xi-{\xi}')}{p_r}\}-\exp\{-\frac{\alpha(\xi-{\xi}'+1)}{p_r}\}\right)$
		\\
		\hline
		$1,\dots,K-1$ & $1,\dots, d_{\min}-1$ & $\zeta$ & $1,\dots,\xi$ & $(1-P_N)\left(\exp\{-\frac{\alpha(\xi-{\xi}')}{p_r}\}-\exp\{-\frac{\alpha(\xi-{\xi}'+1)}{p_r}\}\right)$\\
		\hline
		$1,\dots,K-1$ & $1,\dots,d_{\min}-1$  & $\zeta$ & $d_{\min},\dots,d_{\max}$ & $\frac{P_N}{d_{\max}-d_{\min}+1}\exp\{-\frac{\alpha(\xi)}{p_r}\}$
		\\
		\hline
		$1,\dots,K-1$ & $d_{\min},\dots,d_{\max}$  & $\zeta$ & $1,\dots,d_{\min}-1$ & $(1-P_N)\left(\exp\{-\frac{\alpha(\xi-{\xi}')}{p_r}\}-\exp\{-\frac{\alpha(\xi-{\xi}'+1)}{p_r}\}\right)$
		\\
		\hline
		$1,\dots,K-1$ & $d_{\min},\dots,d_{\max}$  & $\zeta$ & $d_{\min},\dots,\xi$ & $(1-P_N)\left(\exp\{-\frac{\alpha(\xi-{\xi}')}{p_r}\}-\exp\{-\frac{\alpha(\xi-{\xi}'+1)}{p_r}\}\right)
	    \atop+ \frac{P_N}{d_{\max}-d_{\min}+1}\exp\{-\frac{\alpha(\xi)}{p_r}\}$
		\\
		\hline
		$1,\dots,K-1$ & $d_{\min},\dots,d_{\max}$  & $\zeta$ & $\xi+1,\dots,d_{\max}$ & $\frac{P_N}{d_{\max}-d_{\min}+1}\exp\{-\frac{\alpha(\xi)}{p_r}\}$
		\\
		\hline
		$K$ & $1,\dots,d_{\max}$  & $K$ & $1,...,\xi$ & $\exp\{-\frac{\alpha(\xi-{\xi}')}{p_r}\}-\exp\{-\frac{\alpha(\xi-{\xi}'+1)}{p_r}\}$
		\\
		\hline
		$K$ & $1,...,d_{\max}$  & $K-1$ & $d_{\min}, \dots d_{\max}$ & $\frac{1}{d_{\max}-d_{\min}+1}\exp\{-\frac{\alpha(\xi)}{p_r}\}$
		\\
		\hline
	\end{tabular}
\end{table*}

\subsection{Scheduling with Approximate Value Function}
 In this part, we use the value function $ W_{\Pi}(\widetilde {\mathbf S}) $ derived in the previous part to approximate the value function of the optimal policy $ W(\widetilde{\mathbf S}) $ in optimization problem (\ref{eqn:reduced-problem}). Specifically, in the $ t $-th frame, we apply one-step policy iteration based on the approximate value function, and the scheduling actions given system state $\mathbf{S}_t$ can be derived by the following problem.

\begin{Problem}[Sub-optimal Scheduling Problem]\label{prob: sub-optimal}
	\begin{align}
 \Pi'(\mathbf S_t)=\arg\min\limits_{\Omega({\mathbf{S}}_t)}  \bigg\{g'\left(\mathbf{S}_t,\Omega(\mathbf{S}_t)\right)+\sum_{\widetilde{\mathbf{s}}_{t+1}} \gamma \Pr\left(\widetilde{\mathbf{S}}_{t+1} | \mathbf{S}_{t}, \Omega(\mathbf{S}_t)\right) W_{\Pi}(\widetilde{\mathbf{S}}_{t+1}) \bigg\},
	\end{align}
 where $\Pi'$ is the proposed policy after one-step policy iteration from $\Pi$.
	
\end{Problem}

Problem \ref{prob: sub-optimal} can be solved by the following steps.
\begin{itemize}
\item {\bf Step 1:} For each $k\in \mathcal{U}_E(t)$, calculate
\begin{align}
G_E^{k}=	\min\limits_{p_k(t)}  \bigg\{p_k(t)\nonumber + \sum_{\widetilde{\mathbf{S}}_{t+1}} \gamma \Pr\left(\widetilde{\mathbf{S}}_{t+1} | \mathbf{S}_{t}, e_t=1, a_t=k, p_k(t)\right) W_{\Pi}(\widetilde{\mathbf{S}}_{t+1}) \bigg\},\nonumber
\end{align}
and
\begin{align}
	G_L^{k}&=C(n_t)+\min\limits_{p_k(t)}  \bigg\{p_k(t)+ \sum_{\widetilde{\mathbf{S}}_{t+1}} \gamma \Pr\left(\widetilde{\mathbf{S}}_{t+1} | \mathbf{S}_{t}, e_t=0, a_t=k, p_k(t)\right) W_{\Pi}(\widetilde{\mathbf{S}}_{t+1}) \bigg\}.\nonumber
\end{align}
Let $ p_{k,E}^{*} (t)$ and $ p_{k,L}^{*} (t)$ be the optimal power allocation of the above two problems respectively, which can be obtained by one-dimensional search. Note that if there is no arrival of new active device, i.e., $ C(n_t) = 0 $, the above two problems are the same.

\item {\bf Step 2:} If $\min_k G_{E}^k<\min_k G_{L}^k$,  the solution of Problem \ref{prob: sub-optimal} is given by $$\Pi'=\left(e_t=1,a_t=k_E^*,p_k(t)=p_{k_E^*,E}^*(t)\right).$$ where $k_E^*=\arg\min_k G_{E}^k$. Otherwise, the solution of Problem \ref{prob: sub-optimal} is given by	$$\Pi'=\left(e_t=0,a_t=k_L^*,p_k(t)=p_{k_L^*,L}^*(t)\right),$$ where $k_L^*=\arg\min_k G_{L}^k$.
\end{itemize}

The complexity of abovementioned one-step policy iteration
	is $\mathcal{O}(N_p|\mathcal{U}_E(t)|)$, where $N_p$ is the number of quantization levels of transmit power. Moreover, we have the following performance bounds on the proposed scheduling policy.

\begin{Lemma}[Performance Bounds]\label{lem: PerformanceBounds}
	Let $ W_{\Pi'} (\widetilde{\mathbf{S}})\!\triangleq\!\!\lim\limits_{T \rightarrow + \infty }\!\mathbb{E}\bigg[\! \sum_{t=1}^{T} \!\gamma^{t-1} g'\left(\mathbf{S}_t,\Pi'(\mathbf{S}_t)\right)\bigg|\widetilde{\mathbf{S}}_1=\widetilde{\mathbf{S}}\bigg]$ be the value function of the policy $\Pi'$, then
	\begin{align}
		W(\widetilde{\mathbf{S}}) \leq W_{\Pi'}(\widetilde{\mathbf{S}}) \leq	W_{\Pi}(\widetilde{\mathbf{S}}), \forall \widetilde{\mathbf{S}}.
	\end{align}
\end{Lemma}
\begin{proof}
Since policy $\Pi'$ is not the optimal scheduling policy, $W(\widetilde{\mathbf{S}}) \leq	W_{\Pi'}(\widetilde{\mathbf{S}})$ is straightforward. The proof of $W_{\Pi'}(\widetilde{\mathbf{S}}) \leq W_{\Pi}(\widetilde{\mathbf{S}})$ is similar to the proof of \emph{ Policy Improvement Property} in chapter II of \cite{Bertsekas2012Dynamic}.
\end{proof}


\section{Reinforcement Learning Algorithms}
In the previous section, the value function of the baseline scheduling policy $W_{\Pi}(\widetilde{\mathbf{S}})$ is derived analytically. However, the calculation of $W_{\Pi}(\widetilde{\mathbf S})$ requires the priori knowledge on the arrival rate of new active devices and their distribution density, which are usually unknown in practice. Therefore, a novel reinforcement learning approach is proposed in Section \ref{sec:RL} by exploiting the analytical expression of $W_{\Pi}(\widetilde{\mathbf S})$ in equation \eqref{eq:W_Pi_1} and \eqref{eq:W-Pi-3}. Moreover, different values of average receiving power $p_r$ may lead to different performance of both baseline and the proposed policies. Without the statistics of new active devices, it is difficult to optimize $p_r$ directly. Hence, we propose a SGD-based learning algorithm in Section \ref{sec:SGD} to optimize $p_r$ in an online manner. Both online learning algorithms can work simultaneously.
\subsection{Reinforcement Learning for $W_{\Pi}$}\label{sec:RL}
 According to equation \eqref{eq:W_Pi_1} and \eqref{eq:W-Pi-3}, the expression of value function $W_{\Pi}(\widetilde {\mathbf S})$ depends on the arrival rate $P_N$, the expectation of the inverse of pathloss $\varpi=\mathbb{E}_{\rho_{n_t}}[\frac{1}{\rho_{n_t}}]$, and the expected local computing cost for a new active device $\bar C=\mathbb E[C(n_t)]$. $P_N$ may not be known to the BS in advance. Moreover, $\varpi$ and $\bar C$ are the functions of distributions $\lambda$, $\pi_{\ell}$ and $\pi_f$, as defined in (\ref{eq:lambda}), (\ref{eq:pl}) and (\ref{eq:pf}), respectively which may not be known to the BS either. In order to evaluate the value function of the baseline policy, a learning algorithm is proposed below.

\begin{Algorithm}[Reinforcement Learning Algorithm]\label{alg:RL}
\ 
	\begin{itemize}
		\item \textbf{Step 1:} Let $t=0$, $n=0$, initialize $P_N^{(0)}$, $\varpi^{(0)}$ and $\bar C^{(0)}$;
		\item \textbf{Step 2:} In the $t$-th frame, let $t=t+1$  and $I_N^{(t)}$ be the new arrival indicator. Update $P_N^{(t)}$ as follows 
		\begin{align*}
			P_N^{(t)}=\frac{t-1}{t}P_N^{(t-1)}+\frac{1}{t}I_N^{(t)}.
		\end{align*}
		If there is arrival of a new active device, let $n=n+1$. Update $\varpi^{(t)}$ and $\bar C^{(t)}$ as follows
		\begin{align*}
			\varpi^{(t)}=\left\{ \begin{array}{cc}
				\varpi^{(t-1)}, \quad\mbox{if}\ I_N(t)=0,\\
				\frac{n-1}{n}\varpi^{(t-1)}+\frac{1}{n}\frac{1}{\rho^{(t)}}, \quad\mbox{if}\ I_N(t)=1,
			\end{array} \right.
		\end{align*}
		\begin{align*}
			\!\!\!\bar C^{(t)}\!\!=\!\!\!
			\left\{ \begin{array}{cc}
				\bar C^{(t-1)}, \quad\mbox{if}\ I_N(t)=0,\\
				\!\!\!\!\!\frac{\!n\!-\!1\!}{n}\bar C^{(t-\!1)}\!\!+\!\!\frac{1}{n}\!\frac{\sum\limits_{d=d_{min}}^{d_{max}}\!\!\!\!\sum\limits_{\tau=1}^{T_{loc}(d, f^{(t)}, \ell^{(t)})}\!\!\!\!\!\!\!\!\!\!\!\!\!\gamma ^{\tau}[1\!+wp_{loc}(f^{(t)}\!)]}{d_{max}-d_{min}+1}\!, 
				\!\mbox{if}\ \! I_N(\!t\!)\!\!=\!\!1,
			\end{array} \right.
		\end{align*}
		where $\rho^{(t)}$, $f^{(t)}$ and $\ell^{(t)}$ be the pathloss coefficient, CPU frequency and the application-related parameter of the new active device observed in $t$-th frame, respectively;
		\item \textbf{Step 3:} Jump to Step 2, until the iteration converges.
	\end{itemize}
\end{Algorithm}

\begin{Lemma}[Convergence]
	When $t\to \infty$, Algorithm \ref{alg:RL} will converge, i.e., 
	\begin{align}
	\lim_{t\to\infty} P_N^{(t)}&=P_N,\\
	 \lim_{t\to\infty} \varpi^{(t)}&=\varpi ,\\
	 \lim_{t\to\infty}\bar C^{(t)}&=\bar C.
	\end{align}

\end{Lemma}

\begin{proof}
	Note that $P_N$, $\varpi$ and $\bar C$ are updated with their unbiased observation, the convergence is straightforward according to Theorem 1 in Chapter I of \cite{Nikulin1993Unbiased}.
\end{proof}

\begin{Remark}[Learning Efficiency of Algorithm \ref{alg:RL}]
In conventional reinforcement learning algorithms, the value functions need to be evaluated for all state-action pairs (e.g., Q-learning method) or at least a large subset of state-action pairs (e.g., approximate MDP method). Thus, the scheduler needs to traverse a sufficiently large number of states for many times, which results in large computation complexity and slow convergence rate. In our proposed reinforcement learning algorithm, however, we only need to learn some unknown parameters of the value function which are common for all system states. This is because we have the analytical expression of the approximate value function. It is easy to see that the convergence time is significantly shortened.
\end{Remark}
\begin{table*}
	\centering  
	\caption{Non-zero Entries of Matrix $\frac{ \mathrm{d}{\bf \Phi} }{\mathrm{d}p_r}$}  
	\label{table:derivative_Phi}  
		\begin{tabular}{|c|c|c|c|c|}  
			\hline  
			& & & & \\[-6pt]  
			$\zeta$&$\xi$&${\zeta}'$ & ${\xi}'$ & $[
			\frac{ \mathrm{d}{\bf \Phi} }{\mathrm{d}p_r}]_{\epsilon_{\zeta,\xi},\ \epsilon_{{\zeta}',{\xi}'}}$ \\  
			\hline
			1 & $1,\dots,d_{\max}$ & 0& 0& $(1-P_N)\frac{\alpha(\xi)}{p_r^2}\exp\{-\frac{\alpha(\xi)}{p_r}\}$\\
			\hline
			$2,\dots,K-1$ & $1,\dots,d_{\max}$  & $\zeta-1$ & $d_{\min}, \dots, d_{\max}$ & $\frac{1-P_N}{d_{\max}-d_{\min}+1}\frac{\alpha(\xi)}{p_r^2}\exp\{-\frac{\alpha(\xi)}{p_r}\}$
			\\
			\hline
			$1,\dots,K-1$ & $1,\dots,d_{\max}$  & $\zeta+1$ & $1,\dots,\xi$ & $P_N\left(\frac{\alpha(\xi-{\xi}')}{p_r^2}\exp\{-\frac{\alpha(\xi-{\xi}')}{p_r}\}-\frac{\alpha(\xi-{\xi}'+1)}{p_r^2}\exp\{-\frac{\alpha(\xi-{\xi}'+1)}{p_r}\right)$
			\\
			\hline
			$1,\dots,K-1$ & $1,\dots,d_{\min}-1$  & $\zeta$ & $1,\dots,\xi$ & $(1-P_N)\left(\frac{\alpha(\xi-{\xi}')}{p_r^2}\exp\{-\frac{\alpha(\xi-{\xi}')}{p_r}\}-\frac{\alpha(\xi-{\xi}'+1)}{p_r^2}\exp\{-\frac{\alpha(\xi-{\xi}'+1)}{p_r}\}\right)$
			\\
			\hline
			$1,\dots,K-1$ & $1,\dots,d_{\min}-1$ & $\zeta$ & $d_{\min},\dots,d_{\max}$ & $\frac{P_N}{d_{\max}-d_{\min}+1}\frac{\alpha(\xi)}{p_r^2}\exp\{-\frac{\alpha(\xi)}{p_r}\}$\\
			\hline
			$1,\dots,K-1$ & $d_{\min},\dots,d_{\max}$ & $\zeta$ & $1,\dots,d_{\min}-1$ & $(1-P_N)\left(\frac{\alpha(\xi-{\xi}')}{p_r^2}\exp\{-\frac{\alpha(\xi-{\xi}')}{p_r}\}-\frac{\alpha(\xi-{\xi}'+1)}{p_r^2}\exp\{-\frac{\alpha(\xi-{\xi}'+1)}{p_r}\}\right)$\\
			\hline
			$1,\dots,K-1$ & $d_{\min},\dots,d_{\max}$ & $\zeta$ & $d_{\min},\dots,\xi$ & $(1-P_N)\left(\frac{\alpha(\xi-{\xi}')}{p_r^2}\exp\{-\frac{\alpha(\xi-{\xi}')}{p_r}\}-\frac{\alpha(\xi-{\xi}'+1)}{p_r^2}\exp\{-\frac{\alpha(\xi-{\xi}'+1)}{p_r}\}\right)
			\atop + \frac{P_N}{d_{\max}-d_{\min}+1}\frac{\alpha(\xi)}{p_r^2}\exp\{-\frac{\alpha(\xi)}{p_r}\}$\\
			\hline
			$1,\dots,K-1$ & $d_{\min},\dots,d_{\max}$  & $\zeta$ & $\xi+1,\dots,d_{\max}$ & $\frac{P_N}{d_{\max}-d_{\min}+1}\frac{\alpha(\xi)}{p_r^2}\exp\{-\frac{\alpha(\xi)}{p_r}\}$
			\\
			\hline
			$K$ & $1,\dots,d_{\max}$  & $K$ & $1,\dots,\xi$ & $\frac{\alpha(\xi-{\xi}')}{p_r^2}\exp\{-\frac{\alpha(\xi-{\xi}')}{p_r}\}-\frac{\alpha(\xi-{\xi}'+1)}{p_r^2}\exp\{-\frac{\alpha(\xi-{\xi}'+1)}{p_r}\}$
			\\
			\hline
			$K$ & $1,\dots,d_{\max}$  & $K-1$ & $d_{\min}, \dots, d_{\max}$ & $\frac{1}{d_{\max}-d_{\min}+1}\frac{\alpha(\xi)}{p_r^2}\exp\{-\frac{\alpha(\xi)}{p_r}\}$
			\\
			\hline
		\end{tabular}
\end{table*}

\begin{figure}[tb]
	\centering
	\includegraphics[scale=0.6]{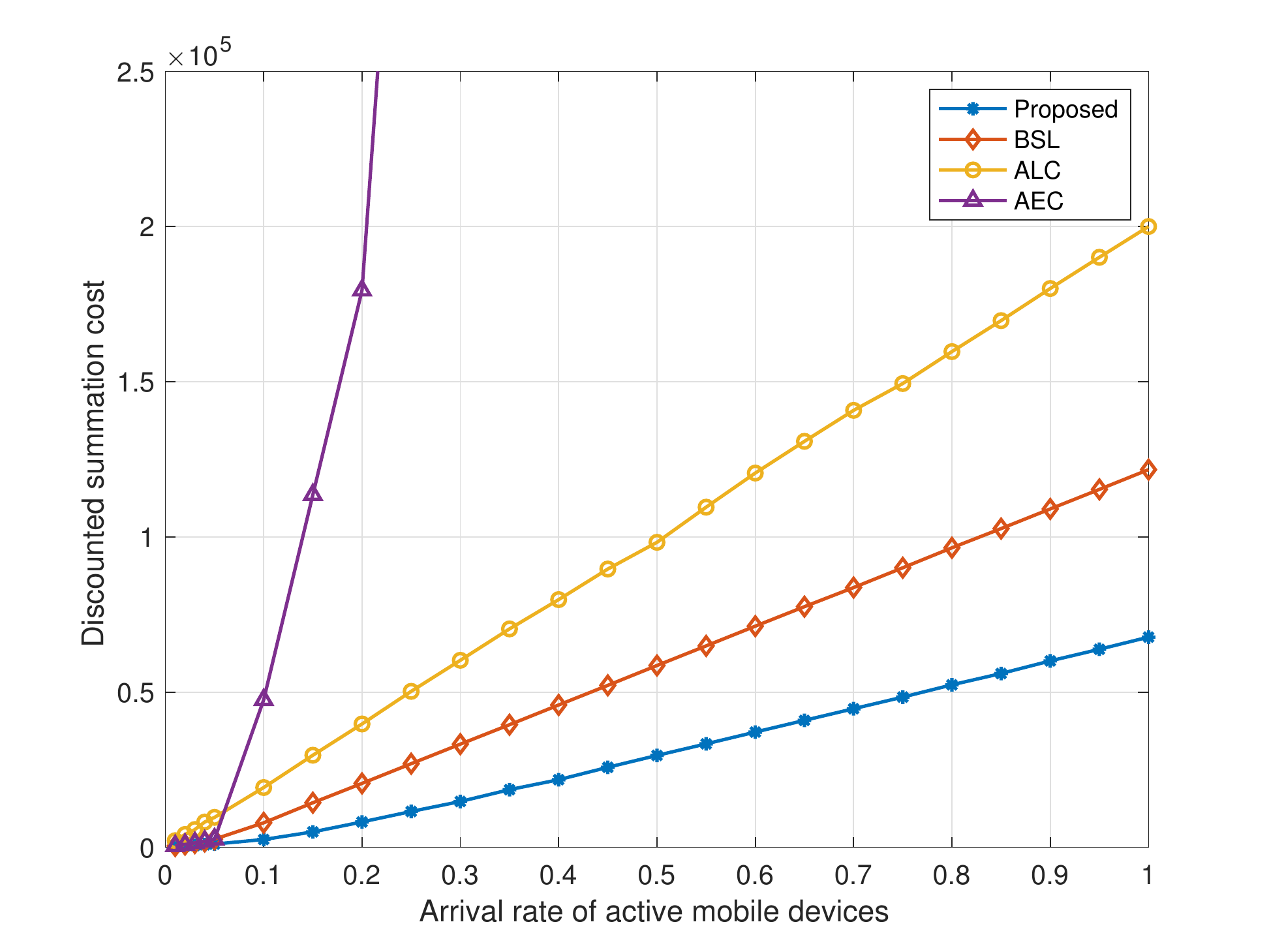}
	\caption{Discounted summation of average system costs versus arrival rate for different policies, where $p_r=2.8\times 10^{-9}$ W, initial system state $\mathbf S_1=(\mathcal U_E(t)=\emptyset, \mathcal U_L(t)=\emptyset, I_N(t)=0)$.}
	\label{fig:cost_discount}
\end{figure}

\begin{figure}[tb]
	\centering
	\includegraphics[scale=0.6]{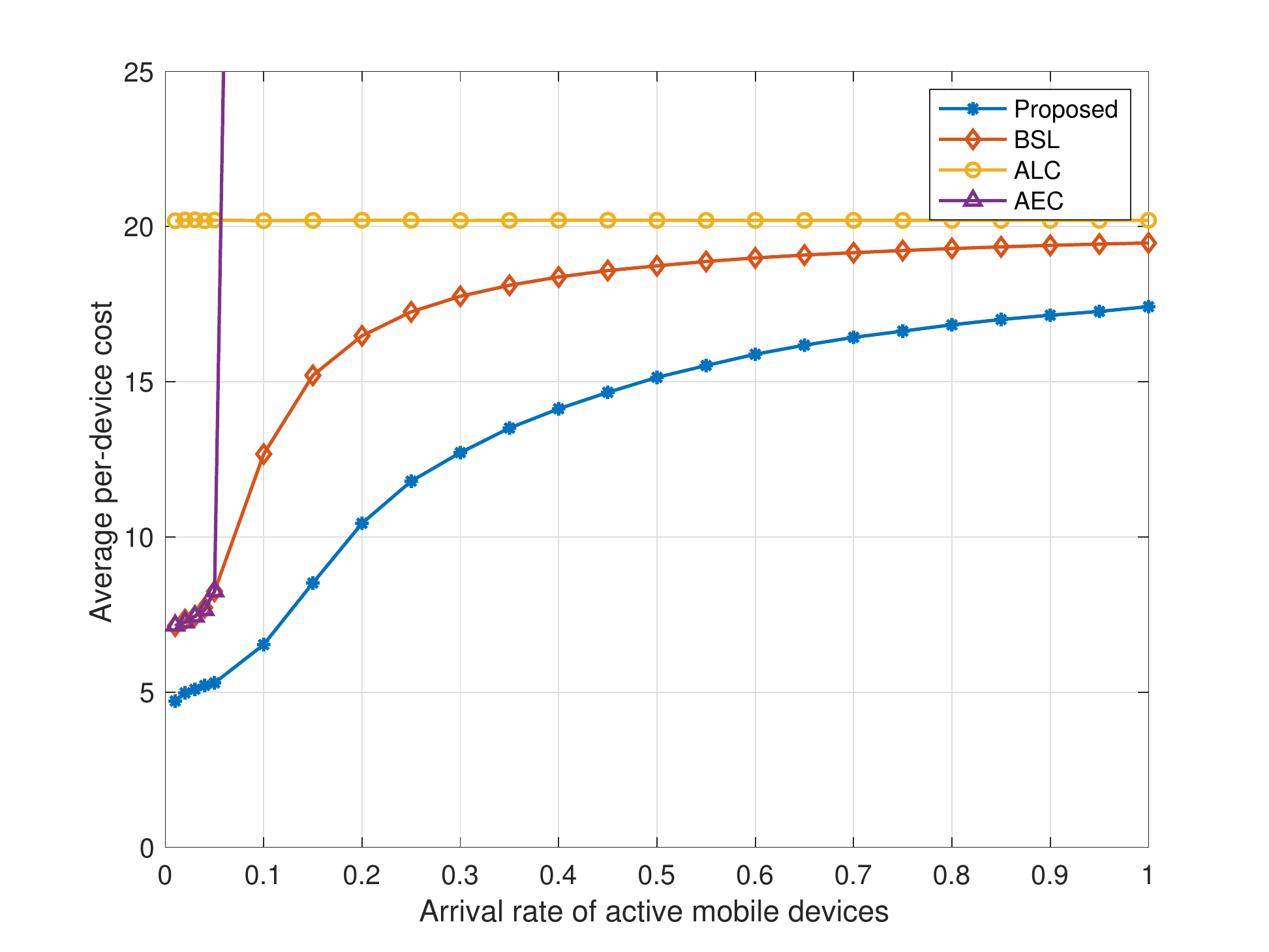}
	\caption{Average per-device costs versus arrival rate for different policies, where $p_r=2.8\times 10^{-9}$ W, initial system state $\mathbf S_1=(\mathcal U_E(t)=\emptyset, \mathcal U_L(t)=\emptyset, I_N(t)=0)$.}
	\label{fig:cost_prob}
\end{figure}

\begin{figure}[tb]
	\centering
	\includegraphics[scale=0.6]{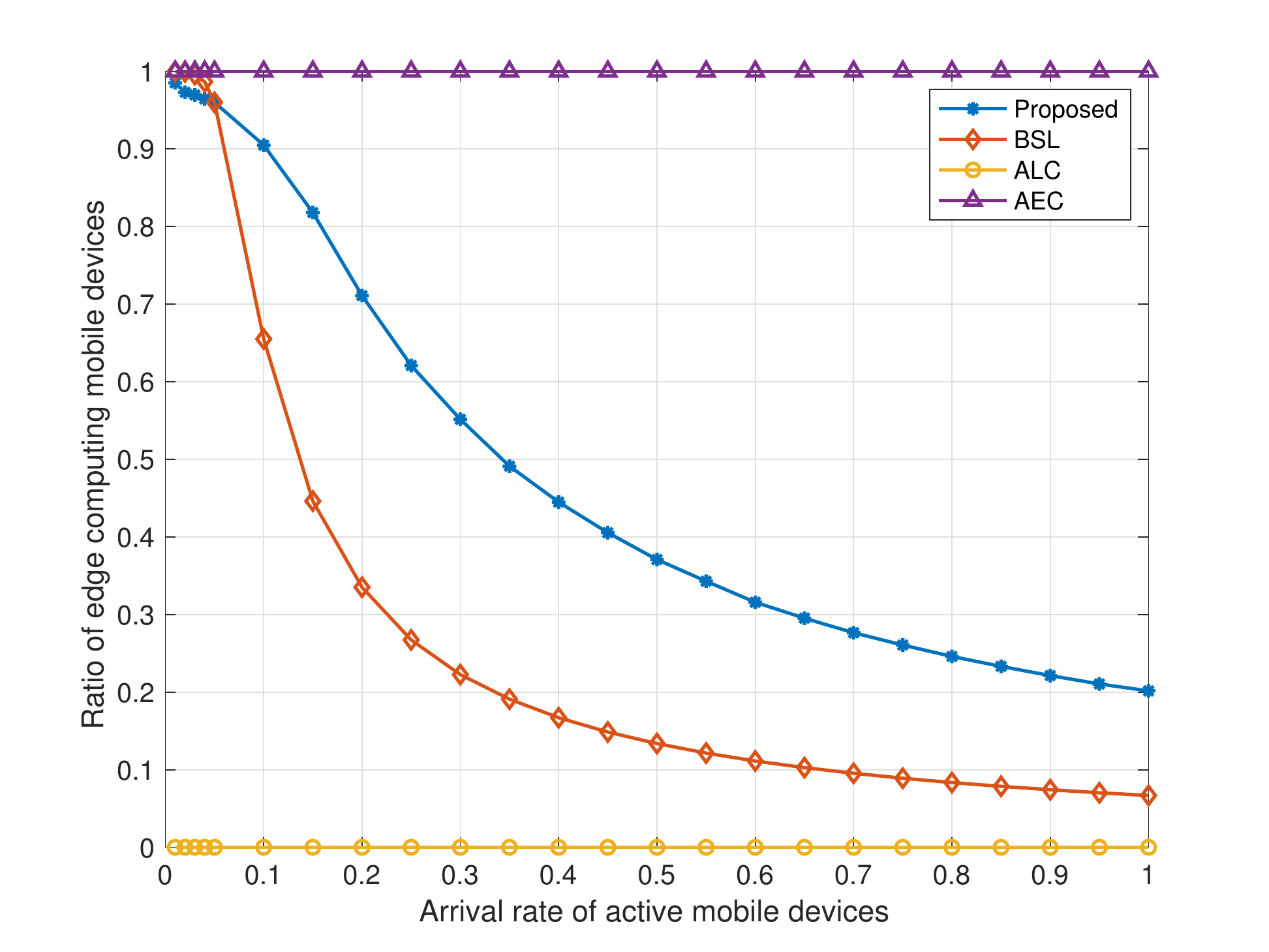}
	\caption{Ratio of edge computing devices versus arrival rate for different policies, where $p_r=2.8\times 10^{-9}$ W, initial system state $\mathbf S_1=(\mathcal U_E(t)=\emptyset, \mathcal U_L(t)=\emptyset, I_N(t)=0)$.}
	\label{fig:ratio_prob}
\end{figure}
\begin{figure}[tb]
	\centering
	\includegraphics[scale=0.6]{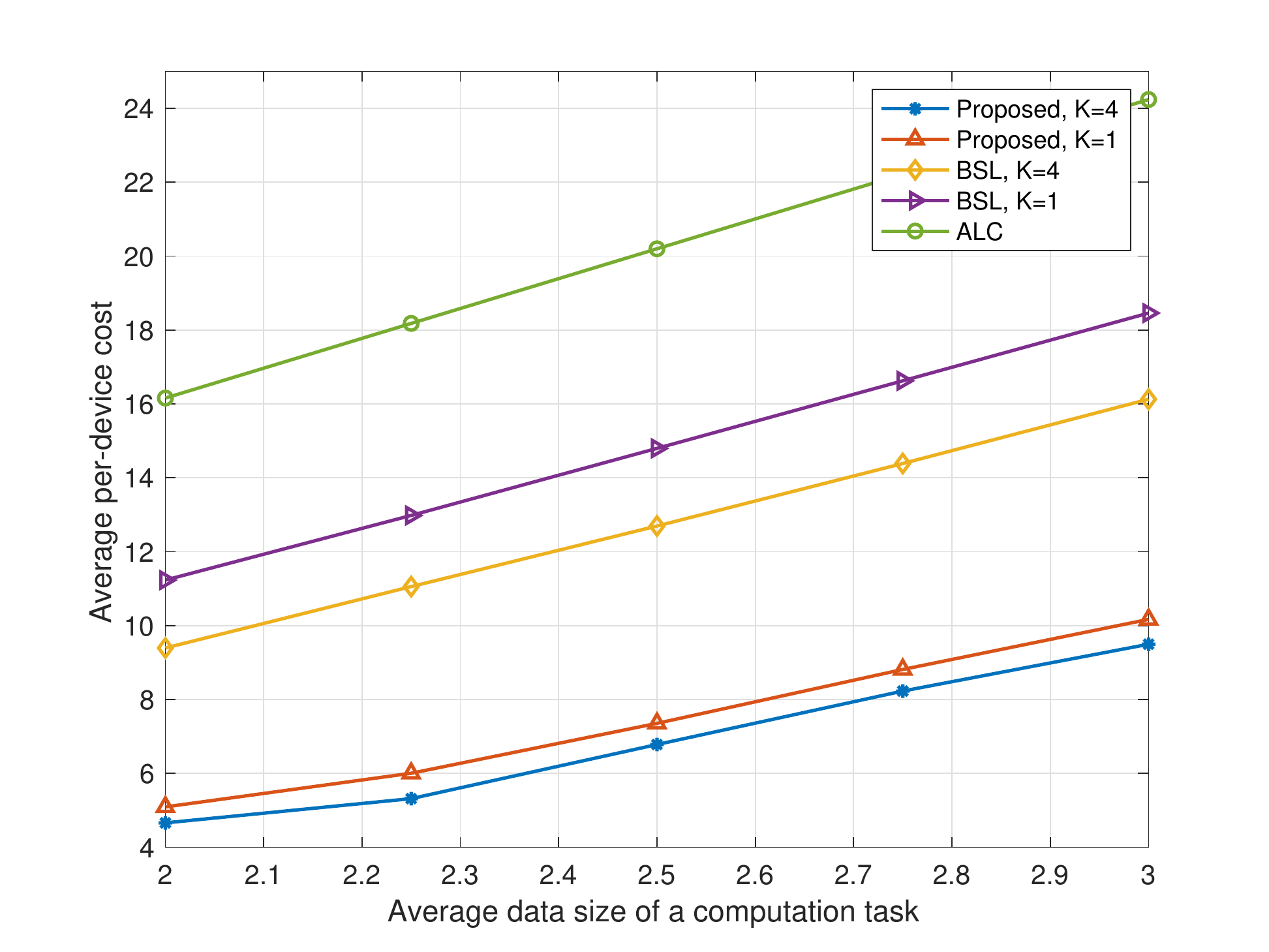}
	\caption{Average per-device costs versus task size for different policies, where $P_N=0.1$, $p_r=2.8\times 10^{-9}$ W, initial system state $\mathbf S_1=(\mathcal U_E(t)=\emptyset, \mathcal U_L(t)=\emptyset, I_N(t)=0)$.}
	\label{fig:cost_size}
\end{figure}

\begin{figure}[tb]
	\centering
	\includegraphics[scale=0.6]{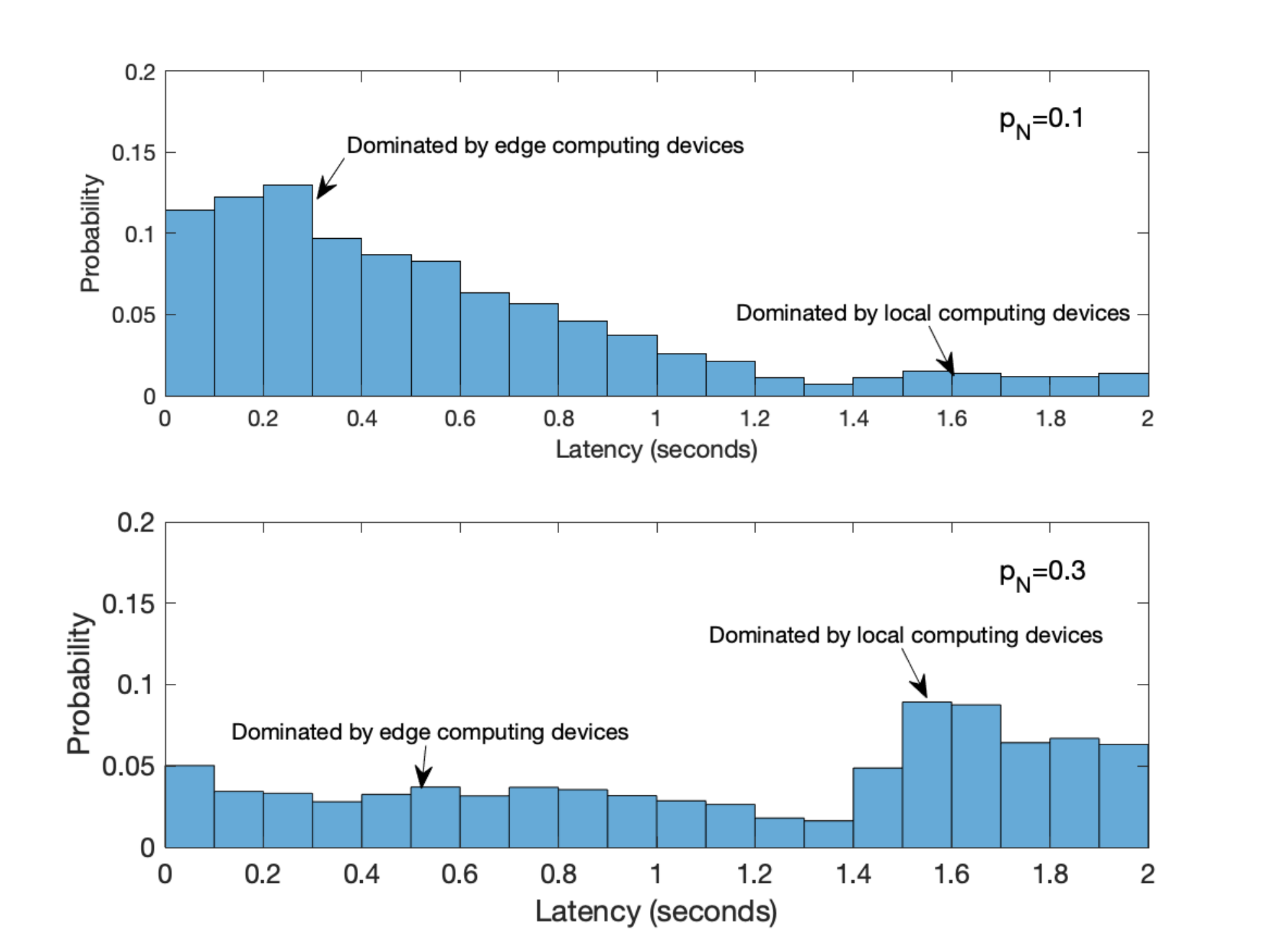}
	\caption{Latency distributions of mobile devices for different arrival rates, where $p_r=2.8\times 10^{-9}$ W, initial system state $\mathbf S_1=(\mathcal U_E(t)=\emptyset, \mathcal U_L(t)=\emptyset, I_N(t)=0)$.}
	\label{fig:latency_dist}
\end{figure}

\begin{figure}[tb]
	\centering
	\includegraphics[scale=0.6]{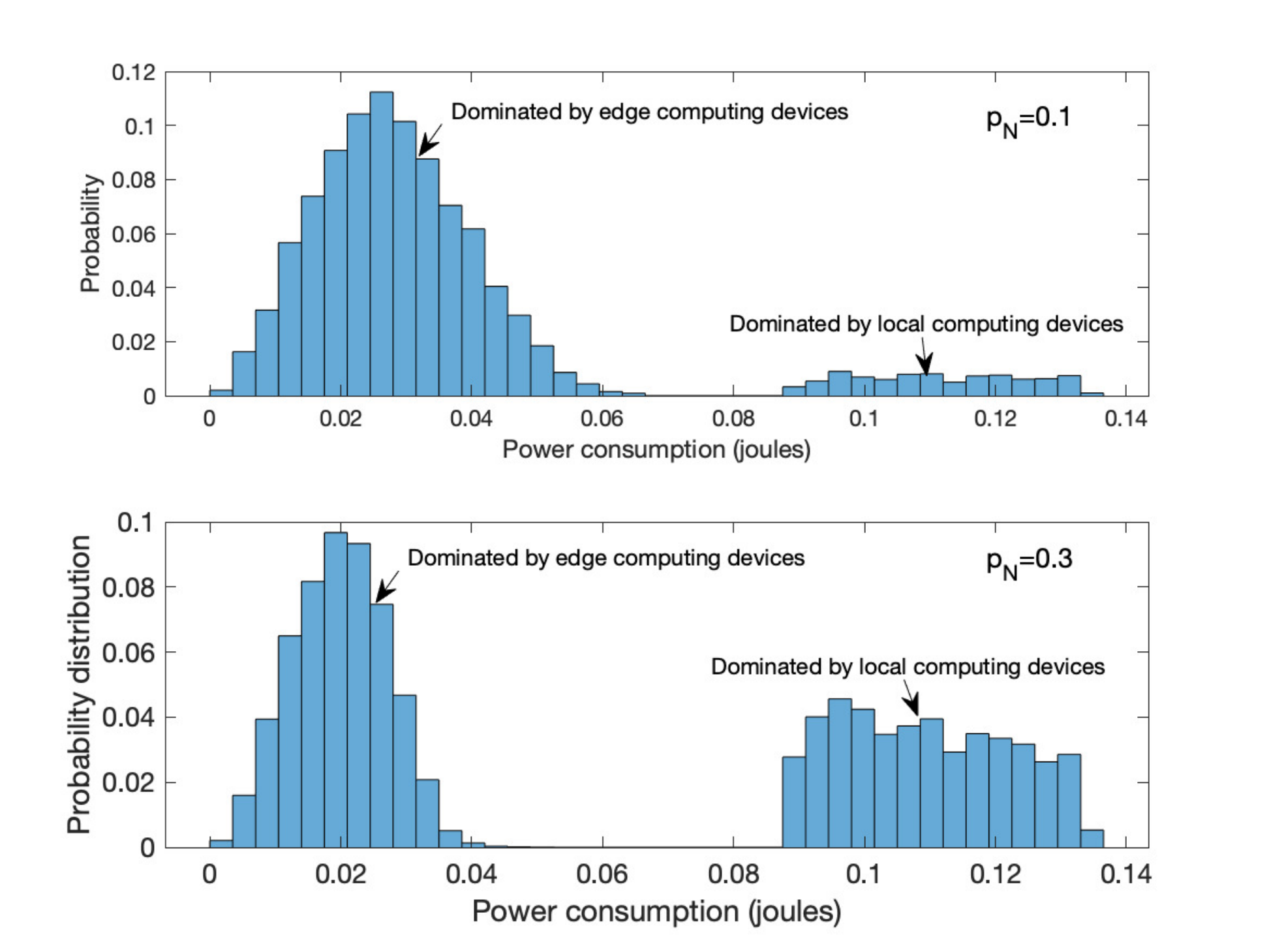}
	\caption{Power consumption distributions of mobile devices for different arrival rates, where $p_r=2.8\times 10^{-9}$ W, initial system state $\mathbf S_1=(\mathcal U_E(t)=\emptyset, \mathcal U_L(t)=\emptyset, I_N(t)=0)$.}
	\label{fig:power_dist}
\end{figure}

\begin{figure}[tb]
	\centering
	\includegraphics[scale=0.6]{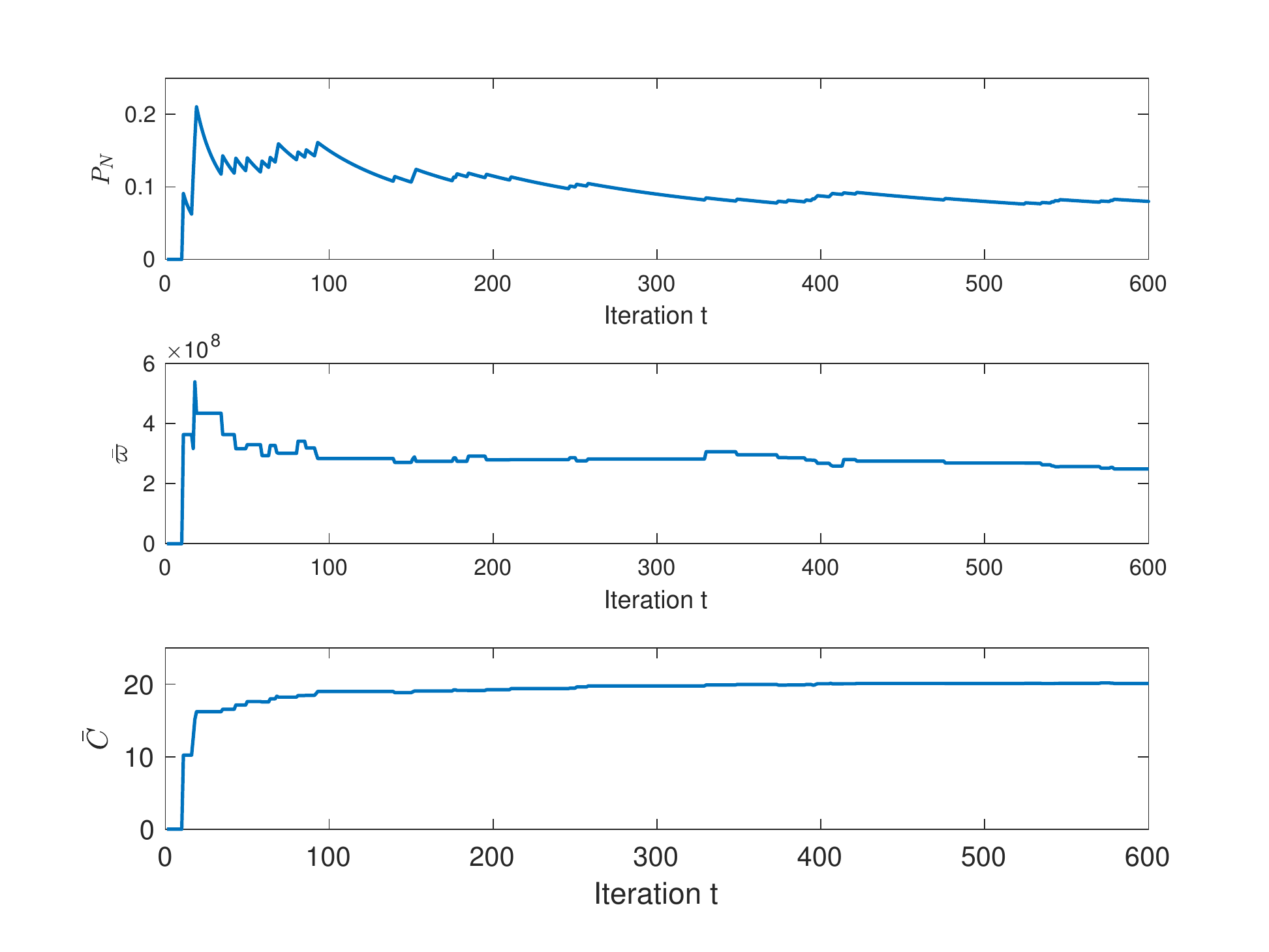}
	\caption{Convergence of the reinforcement learning algorithm.}
	\label{fig:learning_curve}
\end{figure}

\begin{figure}[tb]
	\centering
	\includegraphics[scale=0.6]{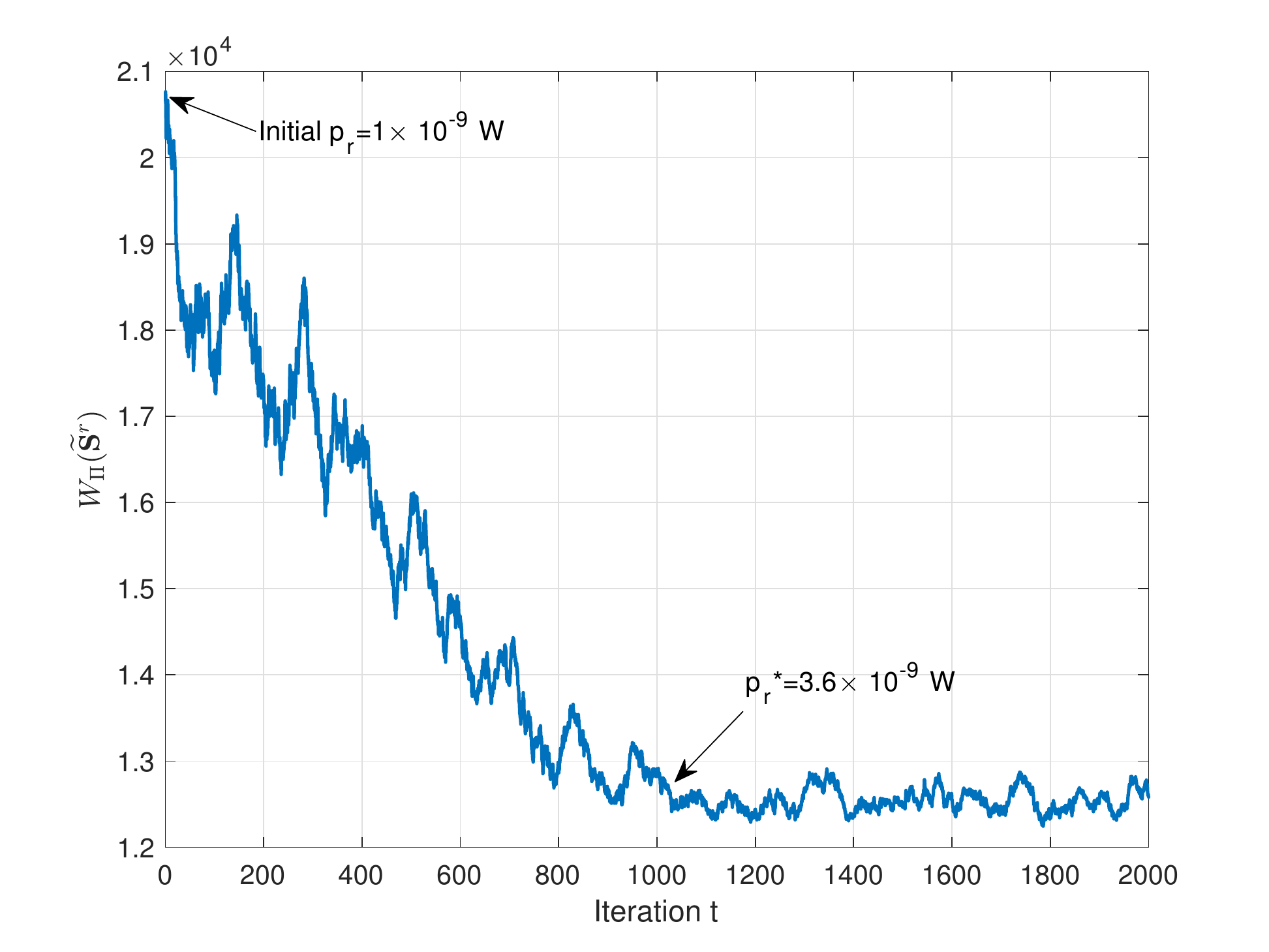}
	\caption{Convergence of the SGD algorithm, where $P_N=0.1$.}
	\label{fig:sgd_curve}
\end{figure}

\begin{figure}[tb]
	\centering
	\includegraphics[scale=0.6]{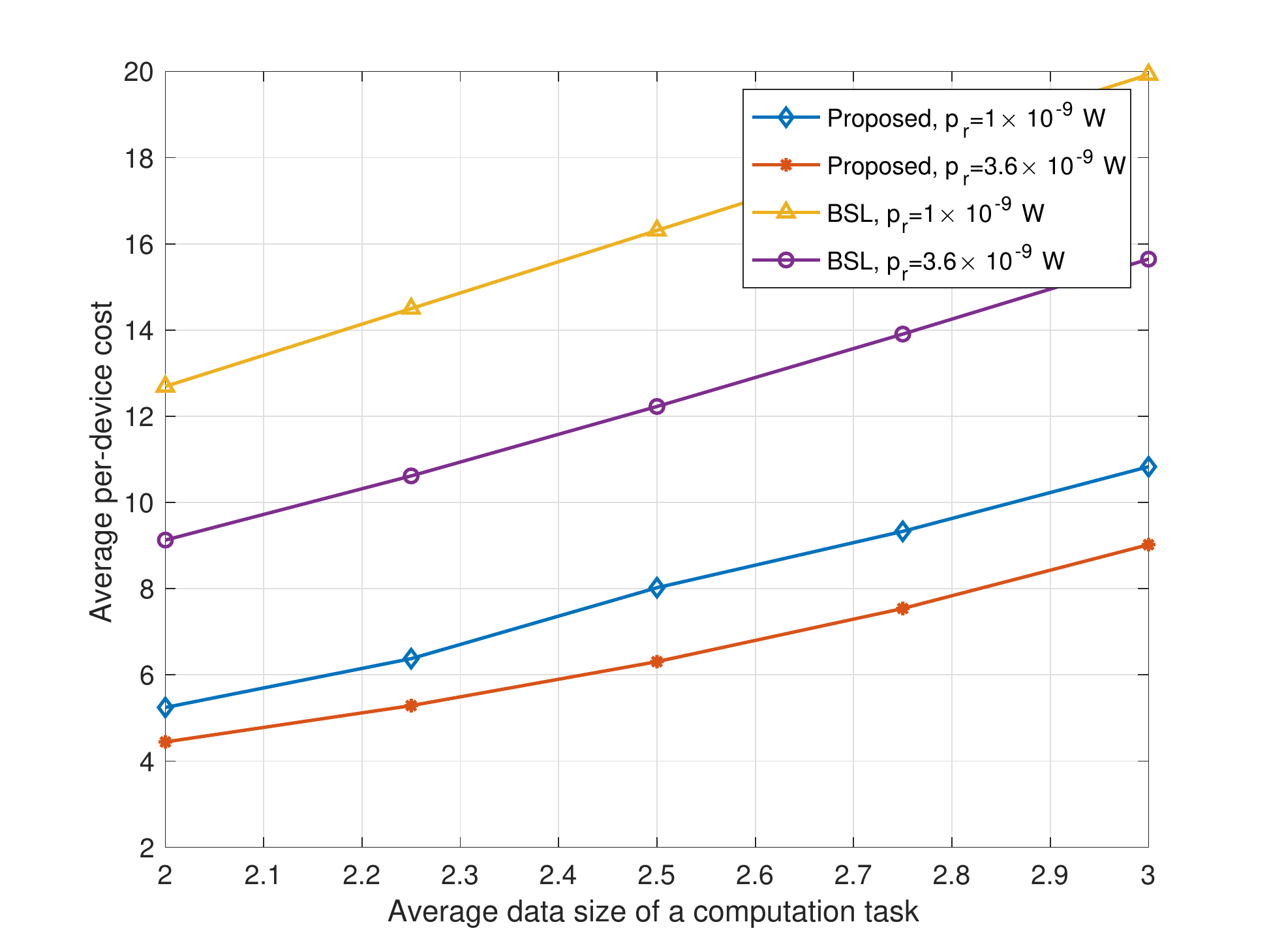}
	\caption{Performance gain of the optimized $p_r^*$, where $P_N=0.1$, $p_r^*=3.6\times 10^{-9}$ W, initial system state $\mathbf S_1=(\mathcal U_E(t)=\emptyset, \mathcal U_L(t)=\emptyset, I_N(t)=0)$.}
	\label{fig:cost_pr}
\end{figure}

\begin{figure}[tb]
	\centering
	\includegraphics[scale=0.6]{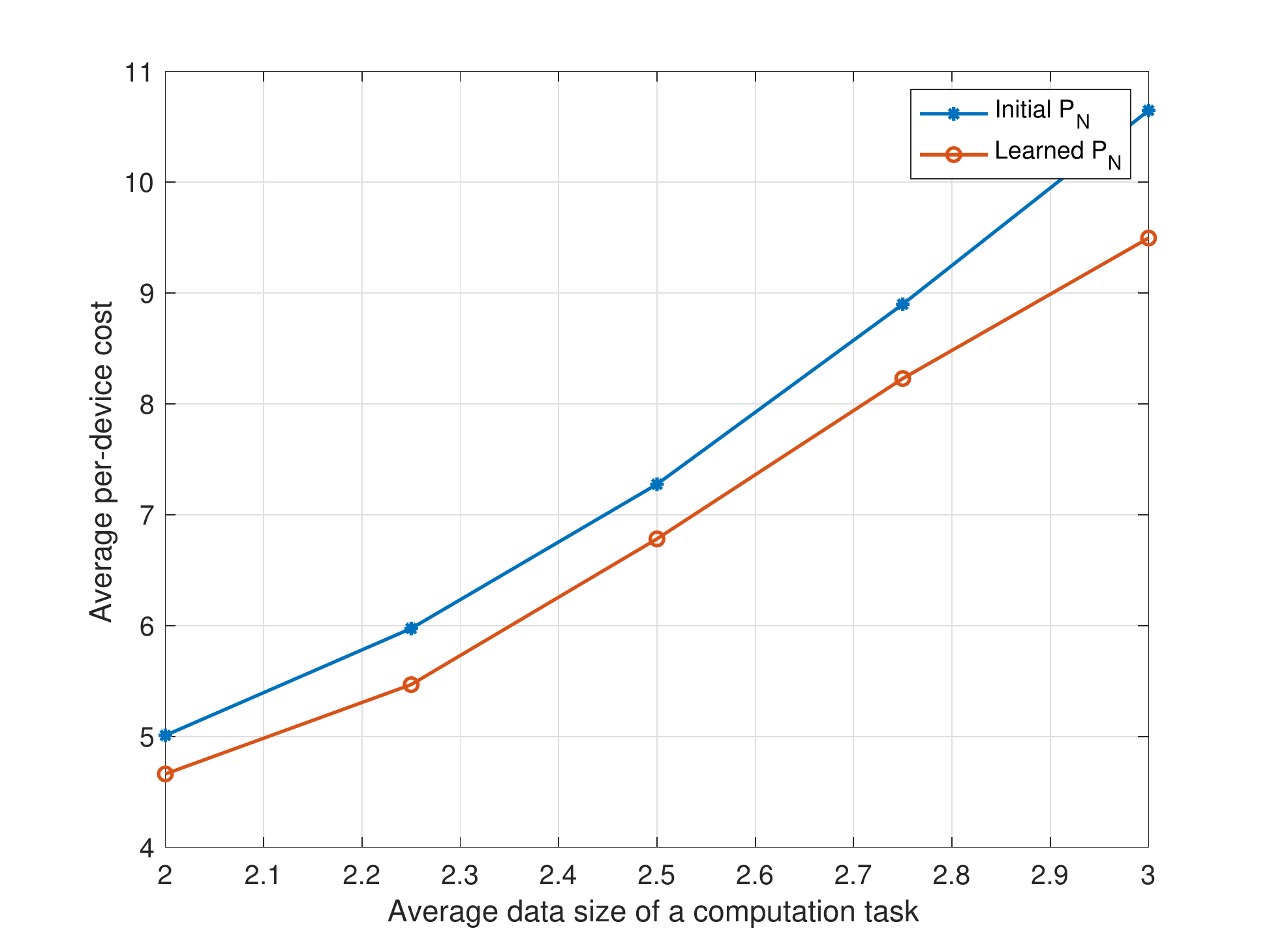}
	\caption{Performance with initial guess $P_N=0.2$ and learned $P_N=0.1$, where $p_r=2.8\times 10^{-9}$ W, initial system state $\mathbf S_1=(\mathcal U_E(t)=\emptyset, \mathcal U_L(t)=\emptyset, I_N(t)=0)$.}
	\label{fig:mismatch}
\end{figure}

\subsection{Optimization of $p_r$ via SGD}\label{sec:SGD}
In this part, we improve the baseline policies by optimizing the average receiving power $p_r$ in the baseline policy. Note that the average system cost is a function of initial system state. It may not be feasible to minimize the average system cost for all the possible initial system states by adjusting $p_r$. Hence, we propose to minimize $W_{\Pi}$ w.r.t. the reference state $\widetilde{\mathbf S}^r$. 
Define the following state without any edge computing device as reference state 
	$$ \widetilde{\mathbf{S}}^{r}\triangleq (\mathcal{U}_E=\emptyset, \mathcal{G}_E=\emptyset, \mathcal{Q}_E=\emptyset).$$

	Then the optimization on  $p_r$ can be written as follows.
	
\begin{Problem}[Optimization of $p_r$]\label{prob:subopt_pr}
\begin{align}
	p_r^*&=\arg\min_{p_r} W_{\Pi}(\widetilde{\mathbf{S}}^r)\nonumber\\
	&=\arg\min_{p_r} \widetilde{\mathbf{v}}^{\mathsf T} (\mathbf{I}-\gamma {\bf \Phi}(p_r))^{-1}\mathbf{c}(p_r),
\end{align}
where $ \widetilde{\mathbf{v}}=[1 \ 0 \ \dots \ 0]^{\mathsf T} \in \mathbb{R}^{(Kd_{\max}+1)\times 1}$.
\end{Problem}

Without the distribution knowledge of $\lambda(\mathbf l)$, a sub-optimal solution of Problem \ref{prob:subopt_pr} can be obtained by the stochastic gradient descent approach. We first introduce the following conclusion on the gradient of $W_{\Pi}(\widetilde{\mathbf{S}}^{r})$ w.r.t. $p_r$.

\begin{Lemma}[Gradient of $W_{\Pi}(\widetilde{\mathbf{S}}^{r})$]\label{lem:gradient}
The derivative of $W_{\Pi}(\widetilde{\mathbf{S}}^{r})$ w.r.t. $p_r$ is given by
\begin{align}\label{eq:derivative_W}
    \frac{\mathrm{d}W_{\Pi}(\widetilde{\mathbf{S}}^{r})}{\mathrm{d}p_r}=\widetilde{\mathbf{v}}^{\mathsf T}(\mathbf{I}-\gamma {\bf \Phi}(p_r))^{-1}\biggl[\frac{\mathrm d\mathbf{c}(p_r)}{\mathrm{d}p_r} 
    +\gamma\frac{\mathrm{d}{\bf \Phi}(p_r)}{\mathrm{d}p_r} (\mathbf{I}-\gamma {\bf \Phi}(p_r))^{-1}\mathbf{c}(p_r)\biggr],
\end{align}
where $\frac{\mathrm{d}{\bf \Phi}(p_r)}{\mathrm{d}p_r}$ is the entry-wise derivative of ${\bf \Phi}(p_r)$ w.r.t. $p_r$. The non-zero entries of $\frac{\mathrm{d}{\bf \Phi}(p_r)}{\mathrm{d}p_r}$ are given in table \ref{table:derivative_Phi}, and other entries are all $0$. Moreover, 
\begin{align}\label{eq:derivative_g}
\frac{\mathrm d\mathbf c(p_r)}{\mathrm dp_r}=\bigg[0, \underbrace{\varpi, \dots, \varpi}_{d_{\max}\  \mbox{items}}\bigg].
\end{align}
\end{Lemma}
\begin{proof}
	Please refer to Appendix C.
\end{proof}

In the gradient expression (\ref{eq:derivative_W}), $\mathbf c(p_r)$ is the function of $\varpi$ which is the expectation of  a function depending on the new active devices' pathloss. Thus, $\varpi$ is unknown in advance. Hence, the following stochastic gradient descent algorithm is developed to optimize $p_r$ together with the learning of $\varpi$.

\begin{Algorithm}[Stochastic Gradient Descent Algorithm]\label{alg:SGD} The stochastic gradient descent algorithm to obtain $p_r^{*}$ is elaborated by the following steps.
\begin{itemize}
    \item \textbf{Step 1:} Let $n=0$, $\varpi^{(0)}$ be the initial value of $\varpi$, and $p_r^{(0)}$ be the initial value of $p_r$ ;
    \item \textbf{Step 2:} If there is arrival of a new active device, let $n=n+1$ and $\rho^{(n)}$ be its pathloss coefficient. Update $\varpi$ and $p_r$ according to 
	\begin{align*}
		\varpi^{(n)}=\frac{n-1}{n}\varpi^{(n-1)}+\frac{1}{n}\frac{1}{\rho^{(n)}},
	\end{align*}
	and
	\begin{align*}
	p_r^{(n)}=&p_r^{(n-1)}-\eta_{n}\widetilde{\mathbf{v}}^{\mathsf T}(\mathbf{I}-\gamma {\bf \Phi}(p_r))^{-1}
	\biggl[[0\  \varpi^{(n)}\  \varpi^{(n)} \dots]^{\mathsf T} \nonumber\\
	&+\gamma\frac{\mathrm{d}{\bf \Phi}(p_r)}{dp_r} (\mathbf{I}-\gamma {\bf \Phi}(p_r))^{-1}[0\  \varpi^{(n)}p_r\  \varpi^{(n)}p_r \dots]^{\mathsf T}\biggr]\bigg|_{p_r^{(n-1)}},&
	\end{align*}

    where $\eta_n$ are step sizes satisfying 
    $$\sum_{n=1}^{\infty}\eta_n=\infty, \quad \sum_{n=1}^{\infty}\eta_n^2<\infty;$$
    \item \textbf{Step 3:} Repeat Step 2, until the iteration converges.
\end{itemize}
\end{Algorithm}

The convergence of the above Algorithm \ref{alg:SGD} follows the standard proof established in \cite{Goodfellow2016DL}, and its performance will be demonstrated in the following section by numerical simulations.

\section{Simulation Results}
In this section, we evaluate the performance of the proposed low-complexity sub-optimal scheduling policy by numerical simulations. In the simulations, the frame duration $T_s=10$ ms. The input data size of each task is uniformly distributed between $ 200 $ and $ 300 $ segments, each with a size of $ 10 $ Kb. Local CPU frequency for each active mobile device is randomly drawn between $0.6\sim1 $ GHz and $ 560\sim600 $ CPU cycles are needed to compute 1-bit input data. The effective switched capacitance is $\kappa=1.2\times 10^{-28}$. Moreover, the cell radius is set to 400 m and the mobile devices are uniformly distributed in the cell region. The uplink transmission bandwidth is $W=10 \ \mbox{MHz} $, noise power is $\sigma_z^2=1\times 10^{-9}$ W and pathloss exponent is 3.5. The weight on latency is set as 0.05. We compare our proposed scheduling policy with three benchmark policies including (1) the \emph{baseline policy} (BSL) as elaborated in section IV-A and we set $K=4$ by default except for Fig. \ref{fig:cost_size} where we compare the performance with $K=1$ and $K=4$; (2) \emph{all local computing policy} (ALC), where all the active devices execute their tasks locally; and (3) \emph{all edge computing policy} (AEC), where all the active devices offload their tasks to the MEC server. The simulation results are shown in the following aspects.

\textbf{Impacts of arrival rate and task size:} Fig. \ref{fig:cost_discount} shows the discounted summation of average system costs versus the arrival rate of active devices for the proposed scheme and different benchmarks. It can be seen that the costs of our proposed scheme and two benchmarks, i.e. BSL and ALC grow approximately linearly with the increase of arrival rate. However, the cost of AEC quickly grows unbounded when the arrival rate becomes large, which is caused by the limited uplink transmission resources. Moreover, the discounted  cost of our proposed scheduling scheme is always significantly lower than the benchmarks, which also demonstrates the cost upper bound derived in Lemma \ref{lem: PerformanceBounds}. Fig. \ref{fig:cost_prob} shows the average per-device costs versus the arrival rates of active devices. It can be observed that the average per-device costs of all the policies grow with the increase of arrival rate except ALC policy. For ALC policy, since all the active devices compute their tasks locally, the arrival rate has no influence on the average per-device cost. For AEC policy, the average per-device cost grows quickly with the increase of arrival rate due to limited wireless transmission capability. It is also shown that our proposed policy always outperforms  BSL policy in average per-device cost especially when the arrival rate falls in the region of $ (0,0.4) $. Besides, it can be seen that when the arrival rate is sufficiently large, the average per-device costs of both BSL policy and our proposed policy converge to the cost of ALC policy. This observation can be explained by Fig. \ref{fig:ratio_prob} which shows that the ratio of edge computing devices tends to $ 0 $ for sufficiently large arrival rate. This is because of the limited uplink transmission bandwidth. Moreover, as shown in Fig. \ref{fig:ratio_prob}, the ratio of edge computing devices of our proposed policy is remarkably lager than that of BSL policy. Hence, our proposed policy can better exploit the MEC server to save the energy consumption of mobile devices and reduce latency. Fig. \ref{fig:cost_size} shows the average per-device costs versus the average input data size of each computation task for different scheduling policies. It can be observed that the average per-device costs grow almost linearly with the increase of task size for different scheduling policies. In comparison, the average per-device costs have different trends in Fig. \ref{fig:cost_prob}, where the per-device cost will saturate for large arrival rate. 

\textbf{Impacts of the parameter $K$ of the baseline policy:}  Fig. \ref{fig:cost_size} also shows the impacts of $K$ on both the baseline policy and our proposed policy, where the curves of average per-device cost versus task size for $K=1$ and $K=4$ are plotted. It can be seen that the baseline policy with $K=4$ can achieve better performance than that with $K=1$. It also results in a lower average per-device cost of our proposed policy. Hence, by properly choosing $K$ of baseline policy, the performance of our scheduling algorithm can be improved.

\textbf{PMFs of latency and power consumption:} Fig. \ref{fig:latency_dist} and Fig. \ref{fig:power_dist} show the PMFs of latency and power consumption of the mobile devices with our proposed scheduling policy, respectively. The left and right parts of the PMFs in Fig. \ref{fig:latency_dist} and Fig. \ref{fig:power_dist} are mainly contributed by edge computing devices and local computing devices, respectively. It can be observed that, for larger arrival rate  $P_N$, the latency of edge computing increases in general, and more devices are scheduled for local computing. Thus, the overall average per-device latency and the average per-device power consumption increase with $P_N$ due to relatively larger local computing latency and power consumption as shown in Fig. \ref{fig:latency_dist} and Fig. \ref{fig:power_dist}.

\textbf{Reinforcement learning and SGD-based optimization of $p_r$:} Fig. \ref{fig:learning_curve} shows the convergence of the reinforcement learning algorithm. It can be seen that the learning processes converge after around 200 observations for $P_N$, $\varpi$ and $\bar C$. The number of observations required for our proposed reinforcement learning algorithm is much smaller than the number of system states in our problem. In contrast, the number of observations required for the convergence of conventional reinforcement learning algorithms is typically much larger than the number of system states. This demonstrates the efficiency of the proposed reinforcement learning algorithm, which benefits from the derived expression of the approximate value function.  The performance of the reinforcement learning is demonstrated in Fig. \ref{fig:mismatch}, where the performance with initial $P_N$ and the learned one are compared. It can be observed that the performance is remarkably improved with the learned value of $P_N$. Fig. \ref{fig:sgd_curve} shows the iteration steps of the SGD algorithm towards a local optimal average receiving power level $p_r^*$. It can be observed that the SGD algorithm converges after around 1000 iterations and the optimized $p_r^*$ is about $3.6\times 10^{-9}$ W. Fig. \ref{fig:cost_pr} shows the performance of both the baseline policy and our proposed scheduling policy can be improved by using the optimized $p_r^*=3.6\times 10^{-9}$ W, compared with its initial value $p_r=10^{-9}$ W. This justifies the necessity of the SGD-based power optimization algorithm.

\section{Concluding Remarks}
In this paper, we formulate the scheduling of a multi-user MEC system with random user arrivals as an infinite-horizon MDP, and jointly optimize the offloading decision, uplink transmission device selection and power allocation. To avoid the curse of dimensionality, we propose a novel low-complexity solution framework to obtain a sub-optimal policy via an analytical approximation of value function. Moreover, to tackle the unknown system statistics in practice, a novel and efficient reinforcement learning algorithm is proposed, and an SGD algorithm is devised to improve both the baseline and the sub-optimal policies. Simulation results demonstrate the significant performance gain of the proposed scheduling policy over various benchmarks.

This work enriches the methodology of approximate MDP for solving resource allocation problems in communication and computing systems. The solution framework proposed in this work can be further applied to the scenarios where the edge computing latency is not negligible. Moreover, it can be generalized to solve many other resource allocation problems with random user arrivals and departures.

\section*{Appendix A: Proof of Lemma \ref{lem: BE_reduced}}
Since the cost of a local computing device $C(n_t)$ is deterministic, we can calculate it immediately and add it to per-stage cost. Thus, the per-stage cost can be expressed as $$g'(\mathbf{S}_t,\Omega(\mathbf S_t))\triangleq w|\mathcal{U}_E(t)| + p(t) + I_N(t)(1-e_t) C(n_t).$$
Hence,
\begin{align*}
\lim\limits_{T \rightarrow + \infty }\mathbb{E} \bigg[\sum_{t=1}^{T} \gamma^{t-1} g'(\mathbf{S}_t, \Omega(\mathbf{S}_t)) \bigg| \mathbf{\widetilde S}_1 \bigg]=\lim\limits_{T \rightarrow + \infty }\mathbb{E}\bigg[\sum_{t=1}^{T} \gamma^{t-1} g(\mathbf{S}_t, \Omega(\mathbf{S}_t)) \bigg| \mathbf{S}_1 \bigg].
\end{align*}
 Then, the Bellman's equations can be rewritten as 
\begin{align*} \label{eq: IntermedianBE}
V(\widehat{\mathbf S}_t)=\min\limits_{\Omega(\mathbf{S}_t)} \bigg[g'\left(\mathbf{S}_t,\Omega(\mathbf{S}_t)\right) + \sum_{\widehat{\mathbf{S}}_{t+1}}\gamma \Pr\left(\widehat{\mathbf{S}}_{t+1} | \mathbf{S}_{t}, \Omega(\mathbf{S}_t)\right) V(\widehat{\mathbf{S}}_{t+1}) \bigg], 
\end{align*}
where $\widehat{\mathbf S}_t \triangleq (\mathbf S^E_t, \mathbf S^N_t)$. Due to the i.i.d. nature of small-scale fading and new arriving devices, we have following Bellman's equation with reduced state space after taking expectation over the above equation.
\begin{align*}
W(\widetilde{\mathbf{S}}_t)&=\mathbb{E}_{\{\mathcal H_E(t), \mathbf S^N_t|\forall t\}}[V(\widehat{\mathbf S}_t)] \nonumber\\
&=\min\limits_{\Omega(\mathbf{S}_t)}\mathbb{E}_{\{\mathcal H_E(t), \mathbf S^N_t|\forall t\}}\bigg\{g'\left( \mathbf{S}_t,\Omega(\mathbf{S}_t)\right)+ \sum_{\widetilde{\mathbf{S}}_{t+1}}\gamma \Pr\left(\widetilde{\mathbf{S}}_{t+1} | \mathbf{S}_{t}, \Omega(\mathbf{S}_t)\right)W (\widetilde{\mathbf{S}}_{t+1}) \bigg\},\nonumber\\
\end{align*}
where $ {\widetilde{\mathbf{S}}_t}\triangleq \mathbf S^E_t/ \mathcal H_E(t)$.

\section*{Appendix B: Proof of Lemma \ref{lem:W1}, \ref{lem:W2} and \ref{lem:W3}}
\subsubsection{ Proof of Lemma \ref{lem:W1}}
The first term of the right-hand-side of equation (\ref{eq:W_Pi_1}) is the expected cost of the active edge computing devices in $\{m_1,m_2,...,m_{\left[|\mathcal{U}_E(1)|-K\right]^{+}}\}$. The second term is the expected total cost of new active devices which arrive before all the mobile devices in $\{m_1,m_2,...,m_{\left[|\mathcal{U}_E(1)|-K\right]^{+}}\}$ finish uplink transmission. All these new arriving devices will be scheduled for local computing. 
Moreover, for sufficiently large input data size, the transmission of one task spans sufficiently many frames. The ergodic channel capacity can be achieved. Hence, equation (\ref{eq:Tk}) holds.

\subsubsection{ Proof of Lemma \ref{lem:W2}}
The first term of the right-hand-side of equation \eqref{eqn:W_2} is the expected power cost of the active edge computing devices in $\{m_{|\left[\mathcal{U}_E(1)|-K\right]^{+}+1},\dots,m_{|\mathcal{U}_E(1)|}\}$. The second term is the sum of expected total delay cost and  total cost of local computing devices arrived during the same time.

\subsubsection{ Proof of Lemma \ref{lem:W3}}
With baseline policy $\Pi$, there are at most $K$ edge computing devices in the period (3).
 In fact, the $ \epsilon_{{\zeta},{\xi}} $-th entry of vector $ \mathbf{v} $ represents the probability that there are $\zeta$ edge computing devices and $ \xi $ segments of the first edge computing device in the uplink transmission queue; 
 the $ \epsilon_{{\zeta},{\xi}}$-th entry of vector $\mathbf c$ is the expected per-stage cost if there are $\zeta$ edge computing devices and $ \xi $ segments of the first edge computing device in the uplink transmission queue;
 the $ (\epsilon_{{\zeta},{\xi}},\epsilon_{{\zeta}',{\xi}'}) $-th entry of matrix $ \bf{\Phi} $ represents the probability that there are ${\zeta}'$ edge computing devices and $ {\xi}' $ segments of the first edge computing device in the uplink transmission queue
  in the next frame, 
given $\zeta$ edge computing devices and $ \xi $ segments of the first edge computing device in the uplink transmission queue in the current frame. Hence, we have the following discussion on $ {\bf\Phi}_{\epsilon_{{\zeta},{\xi}},\epsilon_{{\zeta}',{\xi}'}} $.
\begin{itemize}
	\item {\bf Case 1 (${\zeta}=0, \ {\xi}=0, \ {\zeta}'=0, \ {\xi}'=0$):} Transition from first state ($0$ segment, $0$ edge computing device) to first state means that there is no new active device arrival. Hence  ${\bf\Phi}_{1,1}=1-P_N$.
	
	\item {\bf Case 2 (${\zeta}=0, \ {\xi}=0, \ {\zeta}'=1, \ {\xi}'=d_{\min}, \dots, d_{\max}$):} This means there is a new active device arrival. The probability of a new active device arrival is $P_N$ and the task size of the new active device is uniformly distributed between $d_{\min}$ to $d_{\max}$. Thus, the probability of transiting from first state ($0$ segment, $0$ edge computing device) to $\epsilon_{{\zeta}',{\xi}'}$-th state (${\zeta}'$ edge computing devices, ${\xi}'$ segments in the first edge computing device) for ${\zeta}'=1$  and $ {\xi}'=d_{\min}, \dots, d_{\max}$ is ${\bf\Phi}_{1,\epsilon_{{\zeta}',{\xi}'}} =\frac{P_N}{d_{max}-d_{min}+1}$.

	\item {\bf Case 3 (${\zeta}=1, \ {\xi}=1, \dots, d_{\max}, \ {\zeta}'=0, \ {\xi}'=0$):} This means (i) there is no new active device arrival; (ii) 	the edge computing device transmit ($\xi$) segments within current frame. Hence, we have 
	\begin{align*}
		{\bf\Phi}_{\epsilon_{{\zeta},{\xi}},\epsilon_{{\zeta}',{\xi}'}}
	=(1-P_N)\Pr\left[ \log_2[1+\frac{p_r |h|^2}{\sigma_z^2}]\geq \frac{ \xi b_s}{WT_s}\right]=(1-P_N) \exp\{-\frac{[2^{\xi b_s/(WT_s)}-1]\sigma^{2}_{z}}{p_r} \}. 
	\end{align*}

		\item {\bf Case 4 (${\zeta}=2, \dots, K-1, \ {\xi}=1, \dots, \ d_{\max}, {\zeta}'=\zeta-1, \ {\xi}'=d_{\min}, \dots, d_{\max}$):} This means (i) there is no new active device arrival; (ii)
	the edge computing device transmit $\xi$ segments within current frame; (iii) there are $\xi'$ segments of the second edge computing device. Hence, we have 
	\begin{align*}
		{\bf\Phi}_{\epsilon_{{\zeta},{\xi}},\epsilon_{{\zeta}',{\xi}'}}
	&=\frac{1-P_N}{d_{\max}-d_{\min}+1}\Pr\left[\log_2[1+\frac{p_r |h|^2}{\sigma_z^2}]\geq \frac{\xi b_s}{WT_s} \right] \\ &=\frac{1-P_N}{d_{\max}-d_{\min}+1}\exp\{-\frac{[2^{\xi b_s/(WT_s)}-1]\sigma^{2}_{z}}{p_r}\}.
	\end{align*}

	\item {\bf Case 5 (${\zeta}=1, \dots, K-1, \ {\xi}=1, \dots, d_{\max}, \ {\zeta}'=\zeta+1, \ {\xi}'=1, \dots, \xi$):} This means (i) there is a new active device arrival; (ii)
	the edge computing device transmit ($\xi-{\xi}'$) segments within current frame. Hence, we have
	\begin{align*}
			{\bf\Phi}_{\epsilon_{{\zeta},{\xi}},\epsilon_{{\zeta}',{\xi}'}}
	&=P_N\Pr\left[\frac{(\xi-{\xi}')b_s}{WT_s}\leq \log_2[1+\frac{p_r |h|^2}{\sigma_z^2}]\leq \frac{ (\xi-{\xi}'+1)b_s}{WT_s}\right] \\
 &=P_N\bigg(\exp\{-\frac{[2^{(\xi-{\xi}')b_s/(WT_s)}-1]\sigma^{2}_{z}}{p_r}\}
	-\exp\{-\frac{[2^{(\xi-{\xi}'+1)b_s/(WT_s)}-1]\sigma^{2}_{z}}{p_r}\}\bigg).
	\end{align*}

	\item {\bf Case 6 ( ${\zeta}=1, \dots, K-1, \ {\xi}=1, \dots, d_{\min}-1,  \ {\zeta}'=\zeta, \ {\xi}'=1, \dots, \xi$):} (i) there is no new active device arrival; (ii)
	the edge computing device transmit ($\xi-{\xi}'$) segments within current frame. Hence, we have 
	\begin{align*}
		{\bf\Phi}_{\epsilon_{{\zeta},{\xi}},\epsilon_{{\zeta}',{\xi}'}}
	&=(1-P_N)\Pr\left[\frac{(\xi-{\xi}')b_s}{WT_s}\leq \log_2[1+\frac{p_r |h|^2}{\sigma_z^2}]\leq \frac{ (\xi-{\xi}'+1)b_s}{WT_s}\right] \\
	& =(1-P_N)\bigg(\exp\{-\frac{[2^{(\xi-{\xi}')b_s/(WT_s)}-1]\sigma^{2}_{z}}{p_r}\}
	-\exp\{-\frac{[2^{(\xi-{\xi}'+1)b_s/(WT_s)}-1]\sigma^{2}_{z}}{p_r}\}\bigg).
	\end{align*}

\item {\bf Case 7 (${\zeta}=1, \dots, K-1, \ {\xi}=1, \dots, d_{\min}-1, \ {\zeta}'=\zeta, \ {\xi}'=d_{\min}, \dots, d_{\max}$):} This means (i) there is a new active device arrival; (ii)
the edge computing device transmit $\xi$ segments within current frame; (iii) there are $\xi'$ segments of the second edge computing device. Hence, we have 
\begin{align*}
	{\bf\Phi}_{\epsilon_{{\zeta},{\xi}},\epsilon_{{\zeta}',{\xi}'}}
&=\frac{P_N}{d_{\max}-d_{\min}+1}\Pr\left[\log_2[1+\frac{p_r |h|^2}{\sigma_z^2}]\geq \frac{\xi b_s}{WT_s} \right]\\ &=\frac{P_N}{d_{\max}-d_{\min}+1}\exp\{-\frac{[2^{\xi b_s/(WT_s)}-1]\sigma^{2}_{z}}{p_r}\}.
\end{align*}

\item{\bf Case 8 (${\zeta}=1, \dots, K-1, \ {\xi}=d_{\min}, \dots, d_{\max}, \ {\zeta}'=\zeta, \ {\xi}'=1, \dots, d_{\min}-1$):} (i) there is no new active device arrival; (ii)
the edge computing device transmit ($\xi-{\xi}'$) segments within current frame. Hence, we have 
	\begin{align*}
{\bf\Phi}_{\epsilon_{{\zeta},{\xi}},\epsilon_{{\zeta}',{\xi}'}}
&=(1-P_N)\Pr\left[\frac{(\xi-{\xi}')b_s}{WT_s}\leq \log_2[1+\frac{p_r |h|^2}{\sigma_z^2}]\leq \frac{ (\xi-{\xi}'+1)b_s}{WT_s}\right] \\
& =(1-P_N)\bigg(\exp\{-\frac{[2^{(\xi-{\xi}')b_s/(WT_s)}-1]\sigma^{2}_{z}}{p_r}\}
-\exp\{-\frac{[2^{(\xi-{\xi}'+1)b_s/(WT_s)}-1]\sigma^{2}_{z}}{p_r}\}\bigg).
\end{align*}

	\item {\bf Case 9 (${\zeta}=1, \dots, K-1, \ {\xi}=d_{\min}, \dots, d_{\max}, \ {\zeta}'=\zeta, \ {\xi}'=d_{\min}, \dots, \xi$).} There are two cases: (i) when there is no new active device arrival, 
	the edge computing device transmit ($\xi-{\xi}'$) segments within current frame. (ii) when there is a new active device arrival, 
	the edge computing device transmit $\xi$ segments within current frame and there are $\xi'$ segments of the second edge computing device. 
	Hence, we have 
	\begin{align*}
	{\bf\Phi}_{\epsilon_{{\zeta},{\xi}},\epsilon_{{\zeta}',{\xi}'}}
	&=(1-P_N)\Pr\left[\frac{(\xi-{\xi}')b_s}{WT_s}\leq \log_2[1+\frac{p_r |h|^2}{\sigma_z^2}]\leq \frac{ (\xi-{\xi}'+1)b_s}{WT_s}\right]\\
	&\quad +\frac{P_N}{d_{\max}-d_{\min}+1}\Pr\left[\log_2[1+\frac{p_r |h|^2}{\sigma_z^2}]\geq \frac{\xi b_s}{WT_s} \right]\\
	&  =(1-P_N)\bigg(\exp\{-\frac{[2^{(\xi-{\xi}')b_s/(WT_s)}-1]\sigma^{2}_{z}}{p_r}\}
	-\exp\{-\frac{[2^{(\xi-{\xi}'+1)b_s/(WT_s)}-1]\sigma^{2}_{z}}{p_r}\}\bigg)\\
	& \quad+\frac{P_N}{d_{\max}-d_{\min}+1}\exp\{-\frac{[2^{\xi b_s/(WT_s)}-1]\sigma^{2}_{z}}{p_r}\} 
	\end{align*}

\item {\bf Case 10 (${\zeta}=1, \dots, K-1, \ {\xi}=1, \dots, d_{\max}, \ {\zeta}'=\zeta, \ {\xi}'=\xi+1, \dots, d_{\max}$):} This means (i) there is a new active device arrival; (ii)
the edge computing device transmit $\xi$ segments within current frame; (iii) there are $\xi'$ segments of the second edge computing device. Hence, we have
\begin{align*}
{\bf\Phi}_{\epsilon_{{\zeta},{\xi}},\epsilon_{{\zeta}',{\xi}'}}
&=\frac{P_N}{d_{\max}-d_{\min}+1}\Pr\left[\log_2[1+\frac{p_r |h|^2}{\sigma_z^2}]\geq \frac{\xi b_s}{WT_s} \right]\\ &=\frac{P_N}{d_{\max}-d_{\min}+1}\exp\{-\frac{[2^{\xi b_s/(WT_s)}-1]\sigma^{2}_{z}}{p_r}\}.
\end{align*}

	\item {\bf Case 11 (${\zeta}=K, \ {\xi}=1, \dots, d_{\max}, \ {\zeta}'=K, \ {\xi}'=1, \dots, \xi$):} New arrived device will be scheduled for local computing. Meanwhile, the edge computing device transmit ($\xi-{\xi}'$) segments within current frame. Hence, we have 
\begin{align*}
	{\bf\Phi}_{\epsilon_{{\zeta},{\xi}},\epsilon_{{\zeta}',{\xi}'}}
&=\Pr\left[\frac{(\xi-{\xi}')b_s}{WT_s}\leq \log_2[1+\frac{p_r |h|^2}{\sigma_z^2}]\leq \frac{ (\xi-{\xi}'+1)b_s}{WT_s}\right]\nonumber\\
& =\exp\{-\frac{[2^{(\xi-{\xi}')b_s/(WT_s)}-1]\sigma^{2}_{z}}{p_r}\}
-\exp\{-\frac{[2^{(\xi-{\xi}'+1)b_s/(WT_s)}-1]\sigma^{2}_{z}}{p_r}\}.
\end{align*}

\item {\bf Case 12 (${\zeta}=K, \ {\xi}=1, \dots, d_{\max}, \ {\zeta}'=K-1, \ {\xi}'=d_{\min}, \dots, d_{\max}$):} New arrived device will be scheduled for local computing. Meanwhile, (i) the edge computing device transmit $\xi$ segments within current frame; (ii) there are $\xi'$ segments of the second edge computing device. Hence, we have 
\begin{align*}
	{\bf\Phi}_{\epsilon_{{\zeta},{\xi}},\epsilon_{{\zeta}',{\xi}'}}
&=\frac{1}{d_{\max}-d_{\min}+1}\Pr\left[\log_2[1+\frac{p_r |h|^2}{\sigma_z^2}]\geq \frac{\xi b_s}{WT_s}\right]\\ &=\frac{1}{d_{\max}-d_{\min}+1}\exp\{-\frac{[2^{\xi b_s/(WT_s)}-1]\sigma^{2}_{z}}{p_r}\}.
\end{align*}
\end{itemize}
The entries of the transition probability matrix ${\bf \Phi}$ are given in table \ref{table:Phi}, and other entries are all $0$.

To prove the second equity in equation (\ref{eq:W-Pi-3}), we first show that $||\gamma{\bf \Phi}||<1$, where $||\cdot||$ is the matrix norm. We have $||{\bf \Phi}||=\varrho({\bf \Phi})$ where $\varrho({\bf \Phi})$ is the spectrum radius of ${\bf \Phi}$. Since ${\bf \Phi}$ is a transition probability matrix (stochastic matrix), we have $\varrho({\bf \Phi})=1$ by Perron-Frobenius Theorem \cite{meyer2000matrix}. Also, since the discount factor $\gamma<1$, we have $||\gamma{\bf \Phi}||<1$. Let $\mathbf C_n=\sum_{t=1}^{n}(\gamma{\bf \Phi})^{t-1}$. By dislocation subtraction, we have
$$\mathbf C_n(\mathbf I-\gamma{\bf \Phi})=\mathbf I-(\gamma {\bf \Phi})^{n+1}.$$  
Since $\varrho(\gamma{\bf \Phi})<1$, $(\mathbf I-\gamma{\bf \Phi})$ is nonsingular. Thus,
$$\mathbf C_n=(\mathbf I-\gamma {\bf \Phi})^{-1}-(\gamma {\bf \Phi})^{n+1}(\mathbf I-\gamma {\bf \Phi})^{-1}.$$
Let $n\to\infty$ on both sides of the above equation. We have 
$$\sum_{t=1}^{\infty}(\gamma {\bf \Phi})^{t-1}=(\mathbf I-\gamma{\bf \Phi})^{-1}.$$
Hence, the second equity in equation (\ref{eq:W-Pi-3}) holds.

\section*{Appendix C: Proof of Lemma \ref{lem:gradient}}
The derivation of the gradient of $W_{\Pi}(\widetilde{\mathbf{S}}^{r})$ is given as follows.
\begin{align*}
\frac{\mathrm{d}W_{\Pi}(\widetilde{\mathbf{S}}^{r})}{\mathrm{d}p_r}&=\frac{\mathrm{d}W_{\Pi}^{(3)}(\widetilde{\mathbf{S}}^{r})}{\mathrm{d}p_r}\\
&=\frac{\mathrm{d}\widetilde{\mathbf{v}}^{\mathsf T}(\mathbf I-\gamma {\bf \Phi}(p_r))^{-1}\mathbf{c}(p_r)}{\mathrm{d}p_r}	 \nonumber\\
&=\!\widetilde{\mathbf{v}}^{\mathsf T}\frac{\mathrm{d}(\mathbf I\!-\!\gamma {\bf \Phi} (p_r))^{-1}}{\mathrm{d}p_r}\mathbf c(p_r)+\widetilde{\mathbf{v}}^{\mathsf T}(\mathbf I-\gamma {\bf \Phi}(p_r)\!)^{-1}\!\frac{\mathrm{d}\mathbf c(p_r)}{\mathrm{d}p_r}\nonumber\\
&=-\widetilde{\mathbf{v}}^{\mathsf T}(\mathbf I\!-\!\gamma \! {\bf \Phi}(p_r))^{-1}\frac{\mathrm{d}(\mathbf I-\gamma {\bf \Phi}(p_r\!)\!)}{\mathrm{d}p_r}(\mathbf I\!-\!\gamma {\bf \Phi}(p_r\!)\!)^{-1}\mathbf c(p_r)\nonumber\\
&\quad+\widetilde{\mathbf{v}}^{\mathsf T}(\mathbf I-\gamma {\bf \Phi}(p_r))^{-1}\frac{\mathrm{d}\mathbf c(p_r)}{\mathrm{d}p_r} \nonumber\\
&=\widetilde{\mathbf{v}}^{\mathsf T}(\mathbf{I}-\gamma {\bf \Phi}(p_r))^{-1}\biggl[\frac{\mathrm{d}\mathbf{c}(p_r)}{\mathrm{d}p_r}\nonumber\\
&\quad-\frac{\mathrm{d}(\mathbf{I}-\gamma {\bf \Phi}(p_r))}{\mathrm{d}p_r} (\mathbf{I}-\gamma {\bf \Phi}(p_r))^{-1}\mathbf{c}(p_r)\biggr]\nonumber\\
&=\widetilde{\mathbf{v}}^{\mathsf T}(\mathbf{I}-\gamma {\bf \Phi}(p_r))^{-1}\biggl[\frac{\mathrm{d}\mathbf{c}(p_r)}{\mathrm{d}p_r} \nonumber\\
&\quad+\gamma\frac{\mathrm{d}{\bf \Phi}(p_r)}{\mathrm{d}p_r} (\mathbf{I}-\gamma {\bf \Phi}(p_r))^{-1}\mathbf{c}(p_r)\biggr].
\end{align*}

\bibliographystyle{IEEEtran}
\bibliography{reference_shortened}

\begin{thebibliography}{10}
\providecommand{\url}[1]{#1}
\csname url@samestyle\endcsname
\providecommand{\newblock}{\relax}
\providecommand{\bibinfo}[2]{#2}
\providecommand{\BIBentrySTDinterwordspacing}{\spaceskip=0pt\relax}
\providecommand{\BIBentryALTinterwordstretchfactor}{4}
\providecommand{\BIBentryALTinterwordspacing}{\spaceskip=\fontdimen2\font plus
\BIBentryALTinterwordstretchfactor\fontdimen3\font minus
  \fontdimen4\font\relax}
\providecommand{\BIBforeignlanguage}[2]{{%
\expandafter\ifx\csname l@#1\endcsname\relax
\typeout{** WARNING: IEEEtran.bst: No hyphenation pattern has been}%
\typeout{** loaded for the language `#1'. Using the pattern for}%
\typeout{** the default language instead.}%
\else
\language=\csname l@#1\endcsname
\fi
#2}}
\providecommand{\BIBdecl}{\relax}
\BIBdecl

\bibitem{huang2019mecranduser}
S.~Huang, B.~Lv, and R.~Wang, ``{MDP}-based scheduling design for mobile-edge
  computing systems with random user arrival,'' in \emph{2019 IEEE Global
  Commun. Conf. (GLOBECOM)}, Dec 2019, pp. 1--6.

\bibitem{Abbas2018MECSurvey}
N.~{Abbas}, Y.~{Zhang}, A.~{Taherkordi}, and T.~{Skeie}, ``Mobile edge
  computing: A survey,'' \emph{IEEE Internet Things J.}, vol.~5, no.~1, pp.
  450--465, Feb 2018.

\bibitem{beck2016whitepaper}
\BIBentryALTinterwordspacing
M.~T. Beck, S.~Feld, C.~Linnhoff-Popien, and U.~P{\"u}tzschler, ``Mobile edge
  computing,'' \emph{Informatik-Spektrum}, vol.~39, no.~2, pp. 108--114, Apr
  2016. [Online]. Available: \url{https://doi.org/10.1007/s00287-016-0957-6}
\BIBentrySTDinterwordspacing

\bibitem{You2015SingleUserWPT}
{C. {You}} and K.~{Huang}, ``Wirelessly powered mobile computation offloading:
  Energy savings maximization,'' in \emph{2015 IEEE Global Commun. Conf.
  (GLOBECOM)}, Dec 2015, pp. 1--6.

\bibitem{you2016multiuser}
C.~{You} and K.~{Huang}, ``Multiuser resource allocation for mobile-edge
  computation offloading,'' in \emph{2016 IEEE Global Commun. Conf.
  (GLOBECOM)}, Dec 2016, pp. 1--6.

\bibitem{chen2016gametheorymec}
X.~Chen, L.~Jiao, W.~Li, and X.~Fu, ``Efficient multi-user computation
  offloading for mobile-edge cloud computing,'' \emph{IEEE/ACM Trans. Netw.},
  vol.~24, no.~5, pp. 2795--2808, October 2016.

\bibitem{chen2015decentalizedgame}
X.~Chen, ``Decentralized computation offloading game for mobile cloud
  computing,'' \emph{IEEE Trans. Parallel Distrib. Syst.}, vol.~26, no.~4, pp.
  974--983, April 2015.

\bibitem{huang2012dynamic}
D.~Huang, P.~Wang, and D.~Niyato, ``A dynamic offloading algorithm for mobile
  computing,'' \emph{IEEE Trans. Wireless Commun.}, vol.~11, no.~6, pp.
  1991--1995, June 2012.

\bibitem{mao2016power-delay}
Y.~{Mao}, J.~{Zhang}, S.~H. {Song}, and K.~B. {Letaief}, ``Power-delay tradeoff
  in multi-user mobile-edge computing systems,'' in \emph{2016 IEEE Global
  Commun. Conf. (GLOBECOM)}, Dec 2016, pp. 1--6.

\bibitem{liu2016delayopt}
J.~Liu, Y.~Mao, J.~Zhang, and K.~B. Letaief, ``Delay-optimal computation task
  scheduling for mobile-edge computing systems,'' in \emph{2016 IEEE Int. Symp.
  on Info. Theory (ISIT)}, July 2016, pp. 1451--1455.

\bibitem{Ko2018HetnetMEC}
H.~{Ko}, J.~{Lee}, and S.~{Pack}, ``Spatial and temporal computation offloading
  decision algorithm in edge cloud-enabled heterogeneous networks,'' \emph{IEEE
  Access}, vol.~6, pp. 18\,920--18\,932, 2018.

\bibitem{XQiu2019DRL-MEC}
X.~{Qiu}, L.~{Liu}, W.~{Chen}, Z.~{Hong}, and Z.~{Zheng}, ``Online deep
  reinforcement learning for computation offloading in blockchain-empowered
  mobile edge computing,'' \emph{IEEE Trans. Veh. Technol.}, vol.~68, no.~8,
  pp. 8050--8062, Aug 2019.

\bibitem{LTan2018DRL-Cache-MEC}
L.~T. {Tan} and R.~Q. {Hu}, ``Mobility-aware edge caching and computing in
  vehicle networks: A deep reinforcement learning,'' \emph{IEEE Trans. Veh.
  Technol.}, vol.~67, no.~11, pp. 10\,190--10\,203, Nov 2018.

\bibitem{JWang2019ResorceAlloc-MEC}
J.~{Wang}, L.~{Zhao}, J.~{Liu}, and N.~{Kato}, ``Smart resource allocation for
  mobile edge computing: A deep reinforcement learning approach,'' \emph{IEEE
  Transactions on Emerging Topics in Computing}, pp. 1--1, 2019.

\bibitem{YLiu2019DRL-VehicleEdge}
Y.~{Liu}, H.~{Yu}, S.~{Xie}, and Y.~{Zhang}, ``Deep reinforcement learning for
  offloading and resource allocation in vehicle edge computing and networks,''
  \emph{IEEE Trans. Veh. Technol.}, vol.~68, no.~11, pp. 11\,158--11\,168, Nov
  2019.

\bibitem{YHe2018cache-comp-DRL}
Y.~{He}, N.~{Zhao}, and H.~{Yin}, ``Integrated networking, caching, and
  computing for connected vehicles: A deep reinforcement learning approach,''
  \emph{IEEE Trans. Veh. Technol.}, vol.~67, no.~1, pp. 44--55, Jan 2018.

\bibitem{YCui2010AMDP-OFDM-uplink}
Y.~{Cui} and V.~K.~N. {Lau}, ``Distributive stochastic learning for
  delay-optimal ofdma power and subband allocation,'' \emph{IEEE Trans. Signal
  Process.}, vol.~58, no.~9, pp. 4848--4858, Sep. 2010.

\bibitem{YCui2010AMDP-OFDM-downlink}
V.~K.~N. {Lau} and Y.~{Cui}, ``Delay-optimal power and subcarrier allocation
  for ofdma systems via stochastic approximation,'' \emph{IEEE Trans. Wireless
  Commun.}, vol.~9, no.~1, pp. 227--233, January 2010.

\bibitem{RWang2011DistTowHopMIMO}
{R. {Wang} and V. K. N. {Lau} and Y. {Cui}}, ``Queue-aware distributive
  resource control for delay-sensitive two-hop mimo cooperative systems,''
  \emph{IEEE Trans. Signal Process.}, vol.~59, no.~1, pp. 341--350, Jan 2011.

\bibitem{RWang2013RelayApproxMDP}
R.~{Wang} and V.~K.~N. {Lau}, ``Delay-aware two-hop cooperative relay
  communications via approximate {MDP} and stochastic learning,'' \emph{IEEE
  Trans. Inf. Theory}, vol.~59, no.~11, pp. 7645--7670, Nov 2013.

\bibitem{Han2016}
Z.~{Han}, H.~{Tan}, G.~{Chen}, R.~{Wang}, Y.~{Chen}, and F.~C.~M. {Lau},
  ``Dynamic virtual machine management via approximate markov decision
  process,'' in \emph{2016 IEEE Intl. Conf. on Computer Commun. (INFOCOM)},
  Apr. 2016, pp. 1--9.

\bibitem{Han2018}
Z.~{Han}, H.~{Tan}, R.~{Wang}, S.~{Tang}, and F.~C.~M. {Lau}, ``Online learning
  based uplink scheduling in hetnets with limited backhaul capacity,'' in
  \emph{2018 IEEE Intl. Conf. on Computer Commun. (INFOCOM)}, April 2018, pp.
  2348--2356.

\bibitem{Han2019}
Z.~{Han}, H.~{Tan}, R.~{Wang}, G.~{Chen}, Y.~{Li}, and F.~C.~M. {Lau},
  ``Energy-efficient dynamic virtual machine management in data centers,''
  \emph{IEEE/ACM Trans. Netw.}, vol.~27, no.~1, pp. 344--360, Feb. 2019.

\bibitem{YSun2019PushingCaching}
Y.~{Sun}, Y.~{Cui}, and H.~{Liu}, ``Joint pushing and caching for bandwidth
  utilization maximization in wireless networks,'' \emph{IEEE Trans. Commun.},
  vol.~67, no.~1, pp. 391--404, Jan 2019.

\bibitem{BLv2019Cache}
B.~{Lv}, L.~{Huang}, and R.~{Wang}, ``Joint downlink scheduling for file
  placement and delivery in cache-assisted wireless networks with finite file
  lifetime,'' \emph{IEEE Trans. Commun.}, vol.~67, no.~6, pp. 4177--4192, Jun.
  2019.

\bibitem{Lv2020TCOM}
B.~{Lv}, R.~{Wang}, Y.~{Cui}, Y.~{Gong}, and H.~{Tan}, ``Joint optimization of
  file placement and delivery in cache-assisted wireless networks with limited
  lifetime and cache space,'' \emph{IEEE Trans. Commun.}, pp. 1--1, 2020.

\bibitem{Lv2020JCIN}
B.~{Lyu}, Y.~{Hong}, H.~{Tan}, Z.~{Han}, and R.~{Wang}, ``Cooperative jobs
  dispatching in edge computing network with unpredictable uploading delay,''
  \emph{Journal of Communications and Information Networks}, vol.~5, no.~1, pp.
  75--85, 2020.

\bibitem{KHan2018SpatialModelingMEC}
S.~{Ko}, K.~{Han}, and K.~{Huang}, ``Wireless networks for mobile edge
  computing: Spatial modeling and latency analysis,'' \emph{IEEE Trans.
  Wireless Commun.}, vol.~17, no.~8, pp. 5225--5240, Aug 2018.

\bibitem{mao2016dynamicmec}
Y.~{Mao}, J.~{Zhang}, and K.~B. {Letaief}, ``Dynamic computation offloading for
  mobile-edge computing with energy harvesting devices,'' \emph{IEEE J. Sel.
  Areas Commun.}, vol.~34, no.~12, pp. 3590--3605, Dec 2016.

\bibitem{zhang2013mccstochastic}
W.~{Zhang}, Y.~{Wen}, K.~{Guan}, D.~{Kilper}, H.~{Luo}, and D.~O. {Wu},
  ``Energy-optimal mobile cloud computing under stochastic wireless channel,''
  \emph{IEEE Trans. Wireless Commun.}, vol.~12, no.~9, pp. 4569--4581, Sep.
  2013.

\bibitem{Ji2019WPMEC}
L.~{Ji} and S.~{Guo}, ``Energy-efficient cooperative resource allocation in
  wireless powered mobile edge computing,'' \emph{IEEE Internet Things J.},
  vol.~6, no.~3, pp. 4744--4754, June 2019.

\bibitem{goldsmith_2005}
A.~Goldsmith, \emph{Wireless Communications}.\hskip 1em plus 0.5em minus
  0.4em\relax Cambridge University Press, 2005.

\bibitem{Burd1996Processor}
T.~D. Burd and R.~W. Brodersen, ``Processor design for portable systems,''
  \emph{Journal of VLSI Signal Processing Systems for Signal Image and Video
  Technology}, vol.~13, no. 2-3, pp. 203--221, 1996.

\bibitem{Kleinrock1975QueueSystem}
L.~Kleinrock, \emph{Theory, Volume 1, Queueing Systems}.\hskip 1em plus 0.5em
  minus 0.4em\relax New York, NY, USA: Wiley-Interscience, 1975.

\bibitem{Bertsekas2012Dynamic}
D.~P. {Bertsekas}, \emph{Dynamic Programming and Optimal Control},
  4th~ed.\hskip 1em plus 0.5em minus 0.4em\relax Belmont, MA, USA: Athena
  Scientific, 2012, vol.~2.

\bibitem{WANG2017}
B.~Lv, L.~Huang, and R.~Wang, ``Cellular offloading via downlink cache
  placement,'' in \emph{2018 IEEE Intl. Conf. on Commun. (ICC)}, May 2018, pp.
  1--7.

\bibitem{Ruiwang2018_GLOBECOM}
B.~Lv, R.~Wang, Y.~Cui, and H.~Tan, ``Joint optimization of file placement and
  delivery in cache-assisted wireless networks,'' in \emph{2018 IEEE Global
  Commun. Conf. (GLOBECOM)}, Dec. 2018, pp. 1--7.

\bibitem{Nikulin1993Unbiased}
V.~G. Voinov and M.~S. Nikulin, \emph{Unbiased Estimators and Their
  Applications: Volume 1: Univariate Case}.\hskip 1em plus 0.5em minus
  0.4em\relax Dordercht: Springer Netherlands, 1993.

\bibitem{Goodfellow2016DL}
I.~Goodfellow, Y.~Bengio, and A.~Courville, \emph{Deep Learning}.\hskip 1em
  plus 0.5em minus 0.4em\relax MIT Press, 2016,
  \url{http://www.deeplearningbook.org}.

\bibitem{meyer2000matrix}
C.~D. Meyer, Ed., \emph{Matrix Analysis and Applied Linear Algebra}.\hskip 1em
  plus 0.5em minus 0.4em\relax Philadelphia, PA, USA: Society for Industrial
  and Applied Mathematics, 2000.

\end{thebibliography}

\end{document}